\documentclass{article}
\PassOptionsToPackage{numbers,compress}{natbib}
\usepackage[final,main]{neurips_2026}
\usepackage[utf8]{inputenc}
\usepackage[T1]{fontenc}
\usepackage{graphicx}
\usepackage{braket}
\usepackage{amsmath}
\usepackage{amsthm}
\usepackage{amssymb}
\usepackage{amsfonts}
\usepackage{array}
\usepackage{booktabs}
\usepackage{hyperref}
\usepackage{url}
\usepackage{microtype}
\usepackage{xcolor}
\usepackage{placeins}
\usepackage{wrapfig}
\usepackage{mathtools}

\newcommand{\ind}{\mathbb{I}}
\newcommand{\Sp}{\mathrm{Sp}}
\newtheorem{theorem}{Theorem}[section]
\newtheorem{proposition}[theorem]{Proposition}

\newcommand{\rob}[1]{}

\title{Equivariant Reinforcement Learning for\\ Clifford Quantum Circuit Synthesis}
\author{%
Richie Yeung$^{1,2}$ \quad Aleks Kissinger$^{1}$ \quad Rob Cornish$^{3}$\\
{\normalfont\small $^{1}$Department of Computer Science, University of Oxford \quad $^{2}$Quantinuum}\\
{\normalfont\small $^{3}$College of Computing and Data Science, Nanyang Technological University, Singapore}
}

\begin{document}

\maketitle
\footnotetext[1]{Correspondence to \texttt{richie.yeung@cs.ox.ac.uk}.}

\begin{abstract}
We consider the problem of synthesizing Clifford quantum circuits for devices with all-to-all qubit connectivity.
We approach this task as a reinforcement learning problem in which an agent learns to discover a sequence of elementary Clifford gates that reduces a given symplectic matrix representation of a Clifford circuit to the identity.
This formulation permits a simple learning curriculum based on random walks from the identity.
We introduce a novel neural network architecture that is equivariant to qubit relabelings of the symplectic matrix representation, and which is size-agnostic, allowing a single learned policy to be applied across different qubit counts without circuit splicing or network reparameterization.
On six-qubit Clifford circuits, the largest regime for which optimal references are available, our agent finds circuits within one two-qubit gate of optimality in milliseconds per instance, and finds optimal circuits in 99.2\% of instances within seconds per instance.
After continued training on ten-qubit instances, the agent scales to unseen Clifford tableaus with up to thirty qubits, including targets generated from circuits with over a thousand Clifford gates, where it achieves lower average two-qubit gate counts than Qiskit's Aaronson-Gottesman and greedy Clifford synthesizers.
\end{abstract}

\section{Introduction}

Quantum programs are executed as \emph{circuits}: sequences of elementary gates acting on qubits.
For any given high-level operation of interest, there are usually many different circuits that can implement it.
A quantum compiler finds a circuit that does so while optimizing for some cost function, typically related to the error rates of the gates used in the circuit.

A general quantum operation involving $n$ qubits is described by a $2^n \times 2^n$ complex matrix, and so finding an optimal circuit for a given operation is intractable in general.
As such, practitioners have focussed on important special cases of quantum circuits that admit more concise representations.
In this work, we consider the task of synthesizing \emph{Clifford circuits}, a general and important class of quantum circuits which can be represented tractably by $2n \times 2n$ binary matrices known as \emph{stabilizer tableaus} \cite{AaronsonGottesman2004}.
Figure~\ref{fig:clifford-examples} shows a four-qubit Clifford circuit together with its corresponding tableau representation.
The synthesis task is to implement a target Clifford operation as an explicit circuit using as few two-qubit (or \emph{entangling}) gates as possible, since these are typically much more expensive and error-prone than single-qubit gates \cite{Aseguinolaza2024,LitinskiOppen2018}.

Existing Clifford synthesis methods sit at two ends of a tradeoff between runtime and solution quality.
At one end, there exist polynomial-time algorithms such as Aaronson--Gottesman \cite{AaronsonGottesman2004}, which scale efficiently but often produce many more entangling gates than necessary. 
At the other end, stronger optimization and exact-synthesis methods produce much shorter circuits, but at heavy computational cost.
For example, template-based optimization can take hours per circuit \cite{BravyiTemplates2021}, while the optimal six-qubit database of \cite{Bravyi2022} required $2.1$ TB of storage and over 300,000 CPU-hours to generate. 

In this work, we pursue a middle ground between these two extremes, by learning a reusable neural heuristic for Clifford synthesis that is much lighter than exact or search-heavy optimization, but achieves substantially better entangling-gate counts than polynomial-time baselines.
Our approach also generalises to out-of-distribution targets larger and harder than those seen during training.

We approach this task via reinforcement learning, combining curriculum learning together with a novel neural network architecture that respects the symmetries of the problem.
For the six-qubit benchmark of Bravyi \textit{et al.}\ \cite{BravyiTemplates2021}, which is currently the largest regime where exact references are available, our policy solves every instance, reaches a maximum gap of one CZ gate across the full suite in $21$ seconds, and with extended search matches the exact optimum on $995/1003$ circuits ($99.2\%$) in three hours.
In comparison, the prior state-of-the-art method recovers $982/1003$ optima ($97.9\%$) after 217 hours; across the full reported time-limit sweep, the same 21 unrecovered circuits consume 576 hours without reaching optimum \cite{BravyiTemplates2021,BravyiTemplatesData2021}.

Our neural network architecture is size-agnostic, which allows its reuse across different qubit counts without circuit splicing or network reparameterization.
We exploit this to go beyond the exact regime of \cite{BravyiTemplates2021}.
In particular, after training on six-qubit instances, we continue training the same policy on ten-qubit instances.
The resulting model is able to synthesize Clifford circuits up to 30 qubits while using fewer CZ gates than Qiskit's standard Clifford synthesizers, Aaronson--Gottesman and the Bravyi \textit{et al.}\ greedy method \cite{AaronsonGottesman2004,QiskitCliffordPluginDocs,QiskitCliffordGreedyDocs}.
To our knowledge, our proposal is the first reinforcement learning method for synthesis of fully connected Clifford circuits, and the first learned policy to produce near-optimal synthesis on six-qubit circuits and moreover transfer to much larger targets.

Our paper aims to be accessible to machine learning researchers with no prior background in quantum computing.
In what follows, Section~\ref{sec:background} reviews Clifford synthesis from first principles; Section~\ref{sec:related-work} covers related work; Section~\ref{sec:method} describes our reinforcement-learning formulation and equivariant architecture; Section~\ref{sec:experiments} covers our experiments, and Section~\ref{sec:discussion} discusses limitations.

\begin{figure*}[t]
\centering
\vspace{0.35em}
\makebox[0.88\textwidth][l]{%
\begin{minipage}[c][1.06in][c]{0.61\textwidth}
\centering
\includegraphics[width=\textwidth]{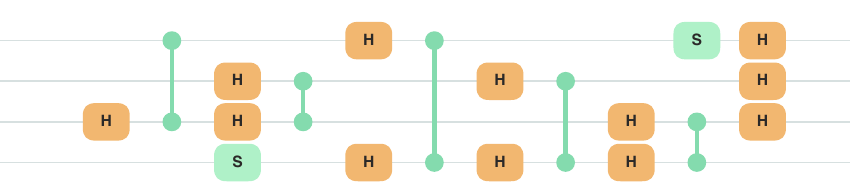}
\end{minipage}%
\hspace{0.03\textwidth}%
\begin{minipage}[c][1.06in][c]{0.26\textwidth}
\centering
\includegraphics[height=1.06in]{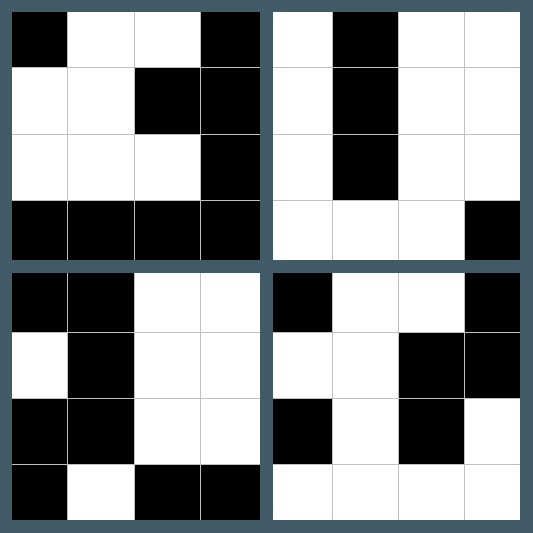}
\end{minipage}}
\caption{A four-qubit Clifford circuit with 5 CZ gates and its corresponding stabilizer tableau, shown as a $2n \times 2n$ binary matrix with black entries for 1 and white entries for 0.
Appendix Fig.~\ref{fig:appendix-clifford-tableau-rollout} shows how successive gate applications locally update this tableau until it reaches the identity.
}
\label{fig:clifford-examples}
\end{figure*}

\section{Background on Clifford Synthesis}
\label{sec:background}

In full generality, a quantum computation involving $n$ qubits requires a $2^n \times 2^n$ unitary matrix to represent.
This becomes intractably large very quickly as $n$ grows.
Accordingly, practitioners have developed more concise representations for important special cases of quantum circuits.
A key example is Clifford circuits \cite{Gottesman1998Heisenberg,AaronsonGottesman2004}.
By the Gottesman-Knill theorem \cite{Gottesman1998Heisenberg}, these correspond to stabilizer tableaus \cite{AaronsonGottesman2004}, whose essential phase-free content can be described succinctly as a $2n \times 2n$ binary symplectic matrix.\footnote{The precise correspondence is described in Appendix~\ref{app:clifford-tableaus}.}
A $2n \times 2n$ binary matrix $M$ is \emph{symplectic} if it satisfies
\[
M^T \Omega M = \Omega,
\qquad
\text{where }
\Omega=
\begin{bmatrix}
0 & I_n \\
I_n & 0
\end{bmatrix}.
\]
Here multiplication is meant over the field $\mathbb{F}_2$ (i.e.\ integers modulo 2), and $I_n$ denotes the $n \times n$ identity matrix.
Below we denote the set of symplectic $2n \times 2n$ binary matrices as $\Sp(2n,\mathbb{F}_2)$.
At a high level, the $(i, j)$-component of each of the four $n \times n$ quadrants of $M$ encodes a certain interaction between the $i$-th and $j$-th qubits of the corresponding Clifford circuit (see \cite{AaronsonGottesman2004} for details).

It is standard to show that $\Sp(2n, \mathbb{F}_2)$ forms a group under matrix multiplication.
Moreover, its group multiplication structure respects the composition of Clifford circuits.
In other words, the Clifford circuit obtained by ``plugging in'' the output of one Clifford circuit into another corresponds to the symplectic matrix obtained by multiplying the tableaus of those two circuits.

The symplectic group admits a collection of generators $\mathcal{G} \subseteq \Sp(2n, \mathbb{F}_2)$, which we denote below as
\begin{equation} \label{eq:clifford-generators}
    \mathcal{G} \coloneqq \underbrace{\{H_i : 1 \leq i \leq n\} \cup \{S_i :1\leq i\leq n\}}_{\text{Single-qubit gates}} \cup \underbrace{\{\mathrm{CZ}_{i,j}:1\leq i < j\leq n\}}_{\text{Two-qubit gates}}.
\end{equation}
In other words, every Clifford circuit corresponds to a product of these generators, and every product of these generators gives a valid Clifford circuit.
Informally, each index $i$ and $j$ refers to a qubit; each $H_i$ is obtained by applying a certain column swap to the identity matrix; and each $S_i$ and $\mathrm{CZ}_{i,j}$ are obtained by applying certain column additions to the identity matrix (see Figure \ref{fig:clifford-generator-actions} for an illustration).
Exact definitions of these generators are given in Appendix \ref{app:clifford-tableaus}.

In physical terms, the generators $H_i$, $S_i$, and $\mathrm{CZ}_{i,j}$ correspond to the native instruction set of a quantum computer.
In order to execute a Clifford circuit on hardware, it must therefore be decomposed into a sequence of these generators.
By using the symplectic matrix representation of Clifford circuits, the basic task of Clifford synthesis therefore becomes as follows:
\begin{equation} \label{eq:clifford-synthesis}
    \text{Given $M_\text{target} \in \Sp(2n, \mathbb{F}_2)$, find $G_1, \ldots, G_k \in \mathcal{G}$ such that $M_\text{target} = G_1 \cdots G_k$}.
\end{equation}
\begin{wrapfigure}{r}{0.50\textwidth}
\vspace{-1.0em}
\centering
\includegraphics[width=\linewidth]{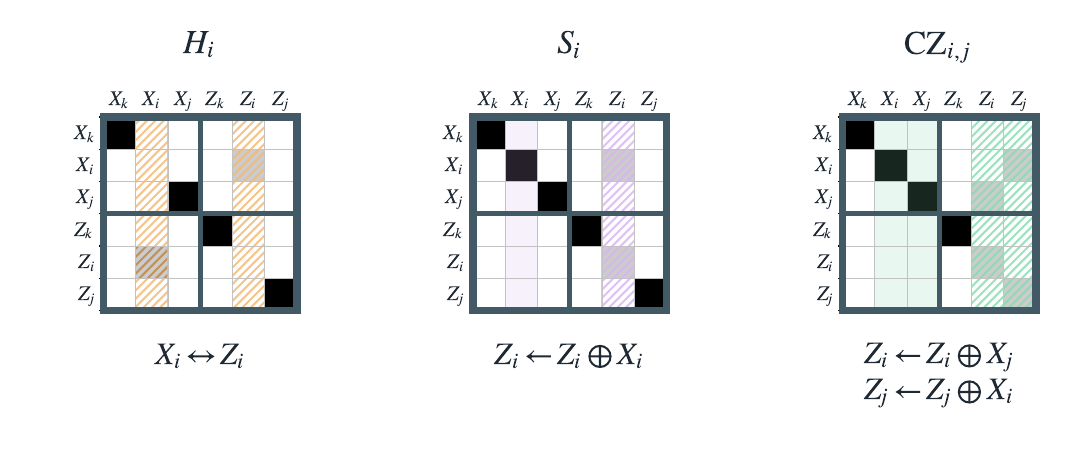}
\vspace{-2.0em}
\caption{Symplectic generator matrices $H_i$, $S_i$, and $\mathrm{CZ}_{i,j}$. Right-multiplying any symplectic matrix by one of these generators applies the highlighted column operations, leaving the other columns untouched.
}
\label{fig:clifford-generator-actions}
\vspace{-1.0em}
\end{wrapfigure}
In general, there are many different sequences $G_1 \cdots G_k$ of different lengths that produce the same overall tableau $M_\text{target}$.
As such, for practical purposes, it is desirable to solve \eqref{eq:clifford-synthesis} in a way that is in some sense efficient.
In particular, it is especially useful to find circuits with fewer two-qubit gates $\mathrm{CZ}_{i,j}$, as these are typically far more error-prone than the single-qubit gates $H_i$ and $S_i$. 
In this way, Clifford synthesis becomes a natural search problem, and therefore amenable to techniques from reinforcement learning, as we explore in this paper.

\section{Related Work}
\label{sec:related-work}
\paragraph{Why Clifford synthesis?}

Clifford operations are a core part of almost every quantum algorithm.
A classical result shows that every quantum computation can be expressed up to arbitrary precision using only Clifford gates ($H$, $S$ and $CZ$) and \emph{$T$ gates} \cite{Boykin2000,Selinger2015}, an additional kind of gate for which two consecutive applications are equivalent to an $S$ gate.
In the context of \emph{fault-tolerant quantum computation}, the implementation of the $T$ gate requires a costly procedure known as \emph{magic state distillation} \cite{BravyiKitaev2005,Litinski2019MagicStateDistillation}. Recent work on \emph{magic state cultivation} \cite{Gidney2024Cultivation} reduces this overhead to roughly match the overhead of a fault-tolerant two-qubit gate such as the $CZ$.
It is therefore practically relevant to reduce $CZ$ gate counts in Clifford circuits as we do here.

\paragraph{Non-neural synthesis.}

Most prior work on Clifford synthesis has focused on traditional, non-neural methods.
These typically involve a sharp trade-off between solution quality and runtime.
At one end of the spectrum, a variety of efficient polynomial-time synthesis procedures have been proposed \cite{AaronsonGottesman2004,BravyiTemplates2021}, which run quickly but often produce circuits far from the optimum.
As an example, Bravyi \textit{et al.} \cite{BravyiTemplates2021} demonstrate that the widely-used Aaronson--Gottesman algorithm produces Clifford circuits with up to 8 times more $CZ$ gates when compared to the output of their peephole optimization procedure, although this procedure can use up to 36 hours per circuit.
On the other hand, various search procedures have also been considered, which can achieve near-optimal synthesis but at a much higher computational cost, which limits their applicability to smaller circuits.
For example, while SAT-based methods can produce certifiably optimal solutions, they have only been applied successfully on 5 qubit Clifford circuits \cite{Peham2023,Shaik2025} and with partial success on 6--7 qubit instances \cite{Shaik2025} using a 3-hour timeout.
A$^\ast$ search has also been used to achieve optimal two-qubit counts for random 5 qubit Clifford tableaus but is outperformed by a greedy heuristic beyond 16 qubits \cite{Webster2025HeuristicOptimal}, both in terms of runtime and two-qubit gate count.

Our work targets the regime between these two extremes: we learn a reusable heuristic that is much lighter than exact or search-heavy optimization, but achieves substantially better entangling-gate counts than polynomial-time baselines.

\paragraph{Neural methods.}

Reinforcement learning has been applied to various related problems in quantum compilation. 
AlphaTensor \cite{Fawzi2022AlphaTensor} was adapted to $T$-count optimization, improving state-of-the-art results on quantum arithmetic circuits \cite{Ruiz2025AlphaTensorQuantum}.
Most closely related to our setting, Kremer \textit{et al.} \cite{kremer2024practical} apply reinforcement learning to the problem of \emph{constrained} Clifford synthesis, a setting where entangling gates can only be applied to nearest-neighbor qubits.
In contrast, we consider the fully unconstrained setting with all-to-all qubit connectivity.
This gives an $O(n^2)$ action space compared to the constrained setting of \cite{kremer2024practical}, whose action space is only $O(n)$.
The general Clifford synthesis task also admits stronger polynomial-time, heuristic, and exact baselines than the constrained setting, making it a more challenging and informative benchmark for learning-based methods.

Beyond these examples, reinforcement learning has also been applied to various tableau synthesis problems, including CNOT synthesis \cite{meijermasters,kremer2024practical,Cossio2026AlphaCNOT}, stabilizer state preparation \cite{Zen2023,Doherty2026GraphDecimation}, and Pauli-network synthesis \cite{Dubal2025}, all of which can be viewed as special cases of our setting with additional constraints on the target Clifford tableau.
Closely related work also applies learning and reinforcement learning to rewriting graphical representations such as ZX diagrams, which encode Clifford and non-Clifford circuit structure \cite{Yeung2023ZXTransformers,Nagele2024ZXRL,Riu2025ZXRL,Mattick2025ZXRL}.
However, current graph-based ZX simplification techniques \cite{Duncan2020GraphTheoreticZX,Kissinger2020ReducingNonClifford} mostly rewrite Clifford structure around fixed non-Clifford gates \cite{Simmons2021PauliFlow, Wetering2025ParametrisedCompilation}; reducing $T$ count requires separate optimization techniques \cite{Amy2016TCountReedMuller,Heyfron2017TODD,DeBeaudrap2020SpiderNest,Ruiz2025AlphaTensorQuantum}.
Compared to these related tasks, our work is the first to demonstrate robust generalization across unseen circuits of different sizes.
In contrast, Kremer \textit{et al.} \cite{kremer2024practical} train 31 models on 31 device architectures of up to 11 qubits for their task.
Since our method improves synthesis directly in the unrouted Clifford setting, the same symmetry-aware policy ideas may also be useful for these neighboring stabilizer-structured compilation problems.

\section{Method}
\label{sec:method}

We approach the problem \eqref{eq:clifford-synthesis} via reinforcement learning rather than supervised learning.
This is due to the combinatorial nature of the problem, which means exact optimal supervision is only available up to the six-qubit regime (i.e.\ the 1003 exact references of \cite{BravyiTemplates2021,Bravyi2022}), whereas we would like to scale beyond this to larger circuit sizes.

\subsection{Reverse-reduction formulation}

We begin by reformulating \eqref{eq:clifford-synthesis} in a way that is better suited to reinforcement learning.
By a standard result, the generators \eqref{eq:clifford-generators} are involutive, so we have $G^{-1} = G$ for each $G \in \mathcal{G}$.
This allows us to write \eqref{eq:clifford-synthesis} equivalently as follows:
\[
    \text{Given $M_\text{target} \in \Sp(2n, \mathbb{F}_2)$, find $G_1, \ldots, G_k \in \mathcal{G}$ such that $M_\text{target} G_1 \cdots G_k = I_{2n}$}.
\]
Given a solution to this latter problem, involutiveness then implies that $M_\text{target} = G_k \cdots G_1$ (i.e.\ the reversed sequence) solves the original problem.
Figure~\ref{fig:clifford-examples} gives an example for a four-qubit tableau, and Appendix Fig.~\ref{fig:appendix-clifford-tableau-rollout} shows the full reduction of that tableau to the identity.

This ``reverse-reduction'' viewpoint is well suited to reinforcement learning because it allows us to define each episode with the same goal state $I_{2n}$, rather than having this be a general tableau that changes across episodes.
As we discuss in Section \ref{sec:curriculum}, it also provides a natural curriculum, where the initial state of each episode is obtained via a random walk away from the identity, and the walk length controls the overall difficulty of that episode.

\subsection{Reinforcement learning setup}

We formulate Clifford synthesis as a deterministic Markov decision process $(\mathcal{S},\mathcal{A},\mathcal{T},\mathcal{R})$, where $\mathcal{S}$ is the state space, $\mathcal{A}$ is the action space, $\mathcal{T}:\mathcal{S}\times\mathcal{A}\to\mathcal{S}$ is the transition function, and $\mathcal{R}:\mathcal{S}\times\mathcal{A}\to\mathbb{R}$ is the reward function.
To simplify notation, we formulate everything with respect to a fixed qubit count that we denote by $n$, although we emphasize that the policy we will train is size-agnostic and can be applied across different qubit counts without reparameterization.

We take the state space to be the set of binary symplectic matrices of size $2n \times 2n$, i.e.\ $\mathcal{S} \coloneqq \Sp(2n,\mathbb{F}_2)$.
The action set is then defined as $\mathcal{A} \coloneqq \mathcal{G}$, where $\mathcal{G}$ is the set of generator matrices from \eqref{eq:clifford-generators}.
There are $\binom{n}{2}$ two-qubit actions and hence $|\mathcal{A}|=\frac{n(n+3)}{2}$ actions in total.
The transition function is
\[
    \mathcal{T}(M,G) \coloneqq M G.
\]
so each step applies one Clifford generator to the current tableau by right multiplication.
In our implementation, an episode terminates whenever the identity matrix $I_{2n}$ is reached, which corresponds to a successful synthesis, or when a step cap is reached, which corresponds to a failed rollout.

Recall that our goal is to synthesize Clifford circuits with as few gates as possible (and in particular, as few two-qubit gates as possible).
Accordingly, we define our reward function as follows:
\[
    \mathcal{R}(M,G) \coloneqq
    -\underbrace{0.01 \, \ind(\text{$G = H_i$ or $S_i$})}_{\text{single-qubit gate penalty}}
    -\underbrace{\ind(G = \mathrm{CZ}_{i,j})}_{\text{two-qubit gate penalty}}
    + \underbrace{25 \, \ind(M G = I_{2n})}_{\text{success bonus}}
    -\underbrace{\frac{\|M G - I_{2n}\|_0}{8n^2}}_{\text{progress shaping}},
\]
where $\ind$ denotes an indicator function and $\|\cdot\|_0$ counts the number of nonzero entries in a matrix. 
The single- and two-qubit gate penalty terms of the reward function correspond to the actual cost of the gates on an actual quantum computer. The success bonus term encourages the agent to avoid the undesirable local optimum of only ever applying single-qubit gates. 
Finally, the last term computes the Hamming distance of the current matrix from the identity, which provides a dense notion of progress throughout an episode.
Other reward weights and shapes may be possible here, but we found empirically that this particular choice is robust and works well in practice.

Overall, our goal is to find a policy $\pi$ to maximise the usual expected sum of discounted rewards over the episode, i.e.\ $\mathbb{E}_\pi[\sum_{t=0}^T \gamma^t R_t]$, where here $R_t$ denotes the reward received at time $t$; the variable $T$ denotes the length of the episode (capped at a maximum number of steps); and $\gamma \in [0,1)$ is a discount factor.
We allow $\pi$ to be stochastic, so that formally $\pi : \mathcal{S} \to \mathcal{P}\mathcal{A}$, where $\mathcal{P}\mathcal{A}$ denotes the set of all probability distributions over $\mathcal{A}$.
In our experiments, we train $\pi$ via PPO \cite{Schulman2017PPO}, and at evaluation time either greedily decode a fixed policy or run policy-guided rollouts (see Appendix~\ref{app:decoding-details}).

\subsection{Clifford synthesis symmetries}

The Clifford synthesis problem admits a natural notion of symmetry.
At a high level, this is because the underlying physical qubits are interchangeable, and so relabeling the qubits in a tableau should correspondingly relabel the qubits in the optimal action sequence.
In turn, these symmetries also translate into symmetries of the reinforcement learning problem itself, as we explain now.

Denote by $\mathrm{Sym}(n)$ the group of permutations of $n$ elements, i.e.\ its elements are bijections of the form $\sigma : \{1, \ldots, n\} \to \{1, \ldots, n\}$.
This group acts on the state space $\mathcal{S}$ as follows:
\[
    \sigma \cdot M \coloneqq \Pi_\sigma M \Pi_\sigma^T,
    \qquad
    \text{where } \Pi_\sigma \coloneqq
    \begin{bmatrix}
    P_\sigma & 0\\
    0 & P_\sigma
    \end{bmatrix},
\]
and $P_\sigma$ is the permutation matrix associated with $\sigma$.
Intuitively, this says that relabeling qubits simultaneously permutes the corresponding rows and columns of the block tableau, while preserving the $2 \times 2$ block structure within each qubit.
$\mathrm{Sym}(n)$ also acts on the action space $\mathcal{A}$ by
\[
    \sigma \cdot H_i \coloneqq H_{\sigma(i)},
    \qquad
    \sigma \cdot S_i \coloneqq S_{\sigma(i)},
    \qquad
    \sigma \cdot \mathrm{CZ}_{i,j} \coloneqq \mathrm{CZ}_{\sigma(i),\sigma(j)}.
\]
Intuitively, this says that after relabeling the qubits, the same physical gate is represented by the correspondingly relabeled qubit indices.
In turn, this leads to an action of $\mathrm{Sym}(n)$ on the space of distributions $\mathcal{P} \mathcal{A}$ via the pushforward (i.e.\ $(\sigma \cdot P)(a) \coloneqq P(\sigma^{-1} \cdot a)$ for $P \in \mathcal{P}\mathcal{A}$).

It is now straightforward to show that the transition function $\mathcal{T}$ and reward function $\mathcal{R}$ are respectively equivariant and invariant to these actions, in the sense that
\[
    \mathcal{T}(\sigma \cdot M, \sigma \cdot G) = \sigma \cdot \mathcal{T}(M,G),
    \qquad
    \mathcal{R}(\sigma \cdot M, \sigma \cdot G) = \mathcal{R}(M,G)
\]
always holds.
By a standard symmetry result for finite MDPs \cite{ZinkevichBalch2001SymmetryMDP}, this leads to the following result. (See also Appendix~\ref{app:equivariant-optimal-policy} for a proof in our notation.)

\begin{proposition}
    \label{prop:optimal-policy}
    There exists an optimal policy $\pi^\star : \mathcal{S} \to \mathcal{P}\mathcal{A}$ whose corresponding value function $V^\star : \mathcal{S} \to \mathbb{R}$ has the following equivariance and invariance properties respectively:
    \[
        \pi^\star(\sigma \cdot M) = \sigma \cdot \pi^\star(M),
        \qquad
        V^\star(\sigma \cdot M) = V^\star(M).
    \]
\end{proposition}

\subsection{Network Architecture}

We solve our reinforcement learning problem using PPO \cite{Schulman2017PPO}, which requires us to parameterize a neural network $\mathcal{S} \to \mathcal{P}\mathcal{A} \times \mathbb{R}$ that maps states to pairs of action distributions and value estimates.
In light of Proposition \ref{prop:optimal-policy}, we would like these outputs to be respectively equivariant and invariant to the symmetries of the problem.
To achieve this, we design a novel neural network architecture that respects these symmetries by construction.
At a high level, the architecture we use is as follows:
\[
\mathcal{S}
\xrightarrow{\text{Embed}}
\underbrace{\mathbb{R}^{n \times n \times h}}_{\text{Edge features}}
\xrightarrow{\text{Aggregate}}
\underbrace{\mathbb{R}^{n \times h}}_{\text{Node features}} \times \underbrace{\mathbb{R}^{n \times n \times h}}_{\substack{\text{Edge features} \\ \text{(unchanged)}}}
\xrightarrow{\substack{\text{Message}\\ \text{passing}}}
\underbrace{\mathbb{R}^{n \times h}}_{\substack{\text{Node features}\\\text{(updated)}}} \times \underbrace{\mathbb{R}^{n \times n \times h}}_{\substack{\text{Edge features} \\ \text{(unchanged)}}}
\xrightarrow{\text{Readout}}
\mathcal{P}\mathcal{A} \times \mathbb{R}
\]
We explain each of the layers in turn now.

\paragraph{Embedding.}

Recall that tableaus are $2n \times 2n$ binary matrices.
Given an input tableau $M_{\text{in}} \in \mathcal{S}$, we begin by reshaping it to an $n \times n \times 4$ binary tensor $M$ according to its four $n \times n$ quadrants (see Figure \ref{fig:clifford-architecture}).
From a physical perspective, each entry $M_{ij} \in \{0, 1\}^4$ then encodes the directed interactions from the $i$-th to the $j$-th qubit in the corresponding Clifford circuit.
We then map each $M_{ij}$ to an embedding in $\mathbb{R}^h$ using an embedding dictionary with $2^5 = 32$ entries, where the dictionary key is given by the four bits in $M_{ij}$ together with a binary indicator $\ind(i = j)$ for whether the entry encodes a self-interaction.
We denote the resulting embeddings by $e_{ij} \in \mathbb{R}^h$.

\paragraph{Aggregation.}

To initialize our message passing, we compute an aggregate embedding $q_i^{(0)} \in \mathbb{R}^h$ for each qubit by mean and max pooling the embeddings $e_{ij}$ over the rows and columns of the tableau.
We give precise details of this pooling in Appendix \ref{app:architecture-details}.
It is also possible to initialise message passing more simply with $q_i^{(0)} = 0$, although we found that aggregation performed better empirically.

\paragraph{Message passing.}

After aggregation, we perform $L$ rounds of message passing, where in each round we compute ordered-pair messages $m_{i \leftarrow j}^{(k)}$ from qubit $j$ to the updated qubit $i$ as follows:
\[
    m_{i \leftarrow j}^{(k)} \coloneqq \phi_{\mathrm{msg}}(q_i^{(k)}, q_j^{(k)}, e_{ij}, e_{ji}, \ind(i \neq j)) \in \mathbb{R}^h,
\]
where $\phi_{\mathrm{msg}}$ is a shared MLP.
Here the indicator $\ind(i \neq j)$ allows the message function to treat diagonal entries differently from off-diagonal ones, which has physical significance since the former encode self-interactions of qubits rather than interactions between distinct qubits.
Given these messages, we then update the qubit embeddings as follows:
\[
    q_i^{(k+1)} \coloneqq \operatorname{LayerNorm}\left(q_i^{(k)} + \phi_{\mathrm{upd}}\left(q_i^{(k)}, \operatorname{mean}_j m_{i \leftarrow j}^{(k)}, \operatorname{max}_j m_{i \leftarrow j}^{(k)}, \eta_{i,1}\right)\right),
\]
where $\eta_{i,1}$ denotes additional \emph{rank-based count features} (see Appendix \ref{app:architecture-details} for a definition).
The latter are motivated by work by \cite{Webster2025HeuristicOptimal}, who show high correlation between these features and the optimal 2-qubit gate count of the related problem of \emph{state preparation}.

\paragraph{Readout.}

After $L$ rounds of message passing, we get final qubit embeddings $q_i \coloneqq q_i^{(L)}$ for each qubit $i$.
We then form a global summary $g \in \mathbb{R}^{h'}$ by computing permutation invariant summary statistics of these final embeddings and their associated count features (see Appendix \ref{app:architecture-details} for details).
To obtain the final action distribution, we then compute
\[
\begin{gathered}
    x_{i} \coloneqq \phi_{\mathrm{H}}(q_i, e_{ii}, g) \in \mathbb{R},
    \qquad
    y_{i} \coloneqq \phi_{\mathrm{S}}(q_i, e_{ii}, g) \in \mathbb{R},\\
    z_{ij} \coloneqq \phi_{\mathrm{CZ}}(p_{ij}, g) \in \mathbb{R},
\end{gathered}
\]
where $\phi_{\mathrm{H}}$, $\phi_{\mathrm{S}}$, and $\phi_{\mathrm{CZ}}$ are MLPs. Here $p_{ij}=p_{ji}$ collects symmetric combinations of the qubit and edge embeddings, defined explicitly in Appendix~\ref{app:architecture-details}.
We then take the collection of every $x_i$, $y_i$, and $z_{ij}$ to be the logits for the policy over $\mathcal{A}$.
This guarantees that the same logit is produced for $\mathrm{CZ}_{i,j}$ and $\mathrm{CZ}_{j,i}$, as required by the symmetry of these actions.
For the value function estimate, a separate MLP acts directly on the permutation-invariant summary statistics defined in Appendix~\ref{app:architecture-details}.

\paragraph{Equivariance and invariance.}

Apart from the value head of the readout layer, each of the individual layers above is equivariant with respect to permutations of the qubits.
For example, it is straightforward to check that relabelling the qubits in the input tableau results in a corresponding relabeling of the embeddings $e_{ij}$, which is then preserved by the aggregation and message-passing layers, and finally results in a corresponding relabeling of the action logits.
Since the value head depends only on permutation-invariant summary statistics, it is itself invariant to qubit relabelling.

\paragraph{Size-agnosticism.}
The same design is also size-agnostic: the only learned weights are in the embedding dictionary and the MLPs, and nothing in the way these are applied depends on the number of qubits $n$.
Accordingly, the same learned weights can be reused across different qubit counts.

\begin{figure*}[t]
\centering
\includegraphics[width=0.98\textwidth]{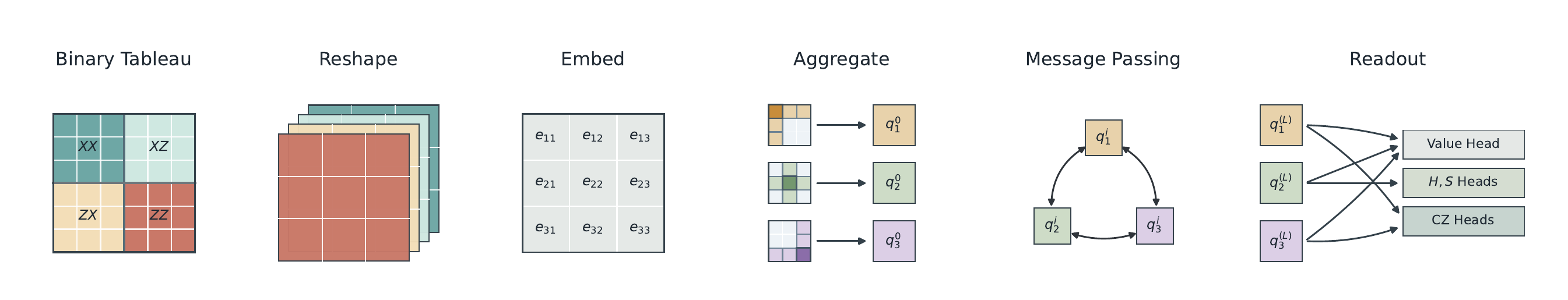}
\par\smallskip
\caption{Architecture of the permutation-equivariant policy used for Clifford synthesis. The input tableau is reshaped and embedded into an $n\times n$ grid of qubit-pair edge features, aggregated into $n$ qubit tokens, updated by edge-conditioned message passing on the complete qubit graph, and passed to equivariant policy heads for $H$, $S$, and $\mathrm{CZ}$ actions, alongside an invariant value head.}
\label{fig:clifford-architecture}
\end{figure*}

\subsection{Curriculum} \label{sec:curriculum}

We have not yet defined how the initial tableau $M_{\mathrm{target}}$ for each episode is generated.
The number of valid tableaus grows as $|\Sp(2n,\mathbb{F}_2)| = 2^{n^2}\prod_{i=1}^n(4^i-1) = 2^{\Theta(n^2)}$ \cite{Bravyi2022}.
As such, if we obtained $M_{\mathrm{target}}$ by sampling uniformly, we would obtain very difficult targets with high probability, so successful episodes would be rare early in training.
We therefore use curriculum learning \cite{Bengio2009CurriculumLearning,Narvekar2020CurriculumSurvey} inspired by reverse curriculum generation for sparse-reward goal-reaching problems \cite{Florensa2017ReverseCurriculum} and DeepCubeA's reverse-from-goal training for combinatorial puzzles \cite{Agostinelli2019DeepCubeA}.
Kremer \textit{et al.} also use curriculum learning for constrained Clifford synthesis \cite{kremer2024practical}.

In our approach, each episode target is generated by a random walk from the identity of some fixed length.
The walk length bounds the distance from the goal and provides a way to control target difficulty.
Training starts with short walks and advances to longer walks once rollout success at the current level reaches 100\%.
Precisely, for $n \in \mathbb{N}$ and $d \geq 0$, we define $\mathcal{P}_{n, d}$ to be the distribution over $\Sp(2n, \mathbb{F}_2)$ generated by a random walk from $I_{2n}$ of length
\[
    L = \lfloor d \rfloor + B,
    \qquad
    \text{where } B \sim \operatorname{Bernoulli}(d-\lfloor d \rfloor).
\]
This means that $\mathbb{E}[L] = d$, and moreover when $d$ is an integer, it reduces to a $d$-step random walk.
This interpolation gives more granular control over the curriculum, especially for small $d$, and shifts the distribution smoothly.
For $n \geq 3$, this random walk converges to uniform sampling from $\Sp(2n,\mathbb{F}_2)$ as $d\to\infty$ (see Theorem~\ref{thm:random-walk-limit} in Appendix~\ref{app:random-walk-limit} for a proof).

In addition to $d$, our size-agnostic agent also allows us to build a curriculum across different qubit counts $n$.
In our experiments, we exploit this by training initially on six-qubit circuits until convergence, before continuing to train on ten-qubit circuits.

\section{Experimental Results}
\label{sec:experiments}

Using the setup described above, we optimized the parameters of our model via PPO \cite{Schulman2017PPO} using PufferLib \cite{Suarez2024PufferLib}.
Full training hyperparameters are listed in Appendix Table~\ref{tab:appendix-training-hparams}.
At evaluation time, we then use our policy as a heuristic for Clifford synthesis, either by greedily following the action with highest probability at each step, or by sampling actions according to the policy until a solution is found or a step cap is reached (see Appendix~\ref{app:decoding-details} for details).

\subsection{Recovering Optimal Six-Qubit Clifford Circuits}

We applied our approach to the benchmark of 1003 optimal six-qubit tableaus considered by Bravyi \textit{et al.} \cite{BravyiTemplates2021}.
(These reference tableaus were themselves taken from the exhaustive database of optimal six-qubit Clifford circuits generated by Bravyi--Latone--Maslov \cite{Bravyi2022}.)
However, we emphasize that our agent is trained according to the methodology described in Section \ref{sec:method}, without any access to those 1003 tableaus or their optimal solutions.
Even the easiest benchmark targets are extremely unlikely to be seen during training, since there are about $1.3\times 10^{14}$ optimal six-qubit Cliffords with five $CZ$ gates \cite{Bravyi2022}.
As a baseline, we compared against the symbolic peephole optimizer of \cite{BravyiTemplates2021}, which matches the optimal entangling-gate count for 982/1003 circuits after 217 hours, with the 21 remaining circuits consuming 576 hours without reaching the optimum \cite{BravyiTemplatesData2021}.

\begin{wrapfigure}{r}{0.40\textwidth}
\vspace{-1.0em}
\centering
\includegraphics[width=\linewidth]{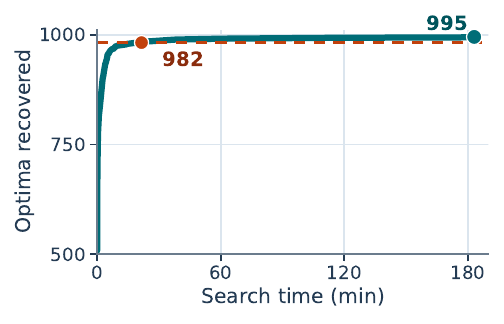}
\vspace{-1.0em}
\end{wrapfigure}

On this task, our approach solves each of the 1003 instances and matches the optimum on 995/1003 circuits ($99.2\%$).
The remaining eight circuits are each off by one CZ gate.
Greedily following the trained policy takes $0.83$ seconds for the whole suite and already recovers 507/1003 optimal solutions.
A brief policy-guided search pass brings all the circuits to within one $CZ$ gate from optimal in $21$ seconds total, with 610/1003 of those being optimal.
The longer policy-guided search reaches the previous state-of-the-art $982/1003$ in $22$ minutes, reaches $990/1003$ after $52$ minutes, and reaches $995/1003$ after $183$ minutes, outperforming the previous state-of-the-art of Bravyi et al.\ \cite{BravyiTemplates2021} by a significant margin.

\subsection{Generalizing Beyond the Training Size}
\label{sec:setup-clifford-transfer}

We also evaluated our approach on larger circuits beyond six qubits.
Since there are no optimal references available in this context, we compare our agent's performance against existing synthesis algorithms that run in polynomial time.
Specifically, we compare against the two main polynomial-time Clifford synthesizers from Qiskit \cite{AaronsonGottesman2004,QiskitCliffordPluginDocs,QiskitCliffordGreedyDocs}, a quantum computing library with over a million monthly downloads at the time of writing.
We do not compare against the neural methods discussed in Section~\ref{sec:related-work}, as they do not address our specific task, and moreover are not size-agnostic and hence would require training separate models for each qubit count in our sweep.

To test size generalization and the cross-size curriculum, we benchmarked two checkpoints: one trained only on six-qubit circuits, and one for which we subsequently continued training on ten-qubit circuits also.
For each qubit size $n$ and number of initial gates $d$, we generated 100 held-out targets by applying $d$ uniformly sampled Clifford gates from $\mathcal{G}$ to the identity tableau.
We swept across $n=7,\ldots,30$ and $d\in\{16,32,64,128,256,512,1024, \infty\}$, using the same fixed target tableaus for the learned agent and both Qiskit baselines.
The $\infty$ circuits are actually sampled uniformly \cite{QiskitRandomCliffordDocs} from $\Sp(2n, \mathbb{F}_2)$ as that is the limit distribution of the random walk (see Theorem \ref{thm:random-walk-limit}).

The results of this experiment are given in Figure~\ref{fig:10fc-difficulty-grid} (see also Appendix~\ref{fig:appendix-10fc-difficulty-grid-full} and \ref{fig:appendix-6fc-difficulty-grid-full} for the full sweeps).
The model trained on ten qubits produces much shorter circuits when it succeeds.
Across the finite-difficulty settings, the model returns lower average CZ counts than the Bravyi \textit{et al.} greedy and Aaronson--Gottesman algorithms.
At 30 qubits and 1024 initial Clifford gates, where the learned synthesizer still solves all targets, it uses $323.3$ CZ gates on average, $124.2$ fewer than Qiskit's Bravyi \textit{et al.} greedy synthesizer and $460.1$ fewer than Aaronson--Gottesman.

At the $\infty$ endpoint, the ten-qubit checkpoint starts to lose reliability beyond about 24 qubits: solve rate falls from $99\%$ at 24 qubits to $59\%$ at 30 qubits. The learned model still beats both Qiskit baselines in CZ count after restricting all methods to its solved targets. The appendix gives the full sweeps (Figures~\ref{fig:appendix-10fc-three-method-cz} and \ref{fig:appendix-10fc-three-method-total-full}).
\begin{figure*}[t]
\centering
\includegraphics[width=0.98\textwidth]{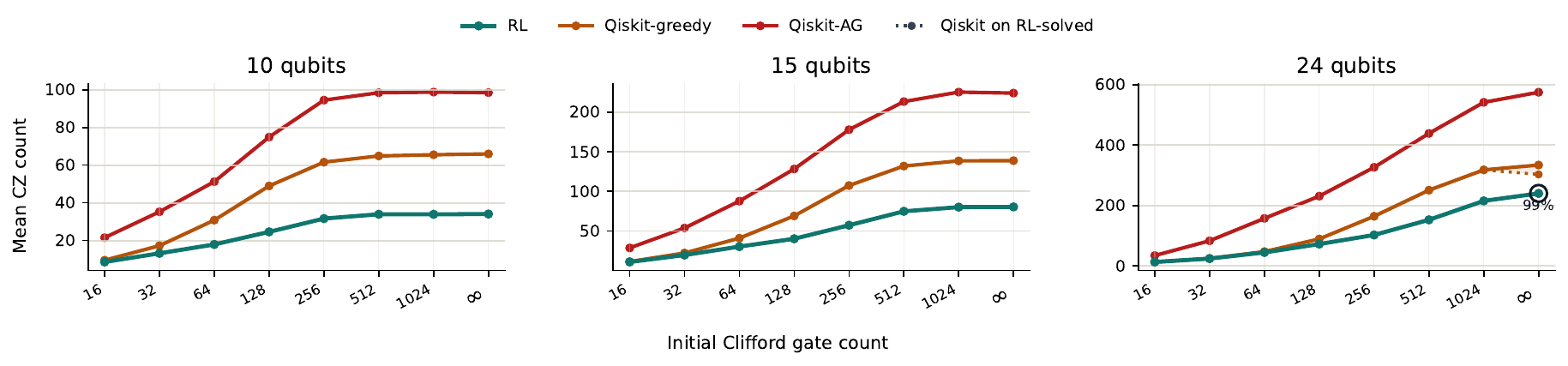}
\caption{CZ-count comparison for the ten-qubit-trained model at 10, 15, and 24 evaluation qubits. Each panel fixes the qubit count and sweeps target difficulty from 16--1024 initial Clifford gates through the $\infty$ endpoint defined in Section~\ref{sec:setup-clifford-transfer}. Solid Qiskit curves show all-target means; dotted same-color markers and shaded gaps show the corresponding Qiskit means restricted to targets solved by the agent when these differ. The percentage annotation gives the agent solve rate at the $\infty$ endpoint. Appendix Figure~\ref{fig:appendix-10fc-difficulty-grid-full} shows the full evaluation on 7--30 qubits.}
\label{fig:10fc-difficulty-grid}
\end{figure*}

Interestingly, these scale-transfer experiments expose a reliability--quality tradeoff between the two checkpoints.
The six-qubit trained model is more reliable: it solves every larger-circuit target in our sweep, including the $\infty$ endpoint.
Its circuits are worse, however, especially on tableaus corresponding to larger and deeper circuits.
See Appendix~\ref{app:sixfc-extrapolation} for additional discussion.

\subsection{Architecture Ablation}
\label{sec:results-clifford-ablation}

We ablate the policy representation by progressively removing qubit-level structure. Ours uses the message-passing architecture from Section~\ref{sec:method}; RelTransformer replaces message passing with relation-aware attention~\cite{Shaw2018RelativeAttention}; Transformer uses ordinary self-attention~\cite{Vaswani2017Attention}; MLP keeps equivariant per-qubit processing but removes communication between qubits; and FlatMLP uses an unconstrained MLP over the flattened tableau.

\begin{table}%
\centering
\small
\begin{tabular}{lc}
\toprule
Family & $CZ$ Count\\
\midrule
Ours & $\mathbf{13.49\pm0.01}$\\
RelTransformer~\cite{Shaw2018RelativeAttention} & $\mathbf{13.47\pm0.06}$\\
Transformer~\cite{Vaswani2017Attention} & $14.12\pm0.28$\\
MLP & $15.88\pm0.05$\\
FlatMLP & $15.19\pm0.16$\\
\bottomrule
\end{tabular}

\vspace{1.0em}
\caption{Architecture ablation on $100$ uniformly sampled six-qubit Clifford targets. Entries are mean CZ counts over solved targets under rollout decoding, reported as mean $\pm$ standard deviation over three seeds. All models solve every target. The full sweep is in Appendix~Table~\ref{tab:appendix-clifford-architecture-ablation}.}
\label{tab:clifford-architecture-ablation}
\end{table}

Table~\ref{tab:clifford-architecture-ablation} shows model performance on uniformly random Cliffords. Since all models solve all targets, the comparison is entirely in circuit quality. The communication-based equivariant models perform best: Ours and RelTransformer are essentially tied, ordinary attention is weaker, removing communication is worse, and the flat unstructured baseline is also substantially behind. The best communication-based models use more than two fewer CZ gates than MLP and about two fewer than FlatMLP. This suggests that the important ingredient is not the exact message-passing update, but the combination of tableau-block features, qubit equivariance, and information flow between qubit tokens.

\FloatBarrier

\section{Discussion, Limitations, and Future Work}
\label{sec:discussion}

We presented a neural synthesis algorithm for Clifford circuits based on reinforcement learning over stabilizer-tableau reduction. The central architectural idea is to build permutation symmetry directly into the policy so that the model can scale across qubit sizes. 
We have shown that this setup is robust, giving optimal circuits at six qubits and improving over existing polynomial-time algorithms for larger qubit systems.

The symmetries and two-dimensional learning curricula we study here exist in related problems, and it is an open question whether these techniques developed in this work would yield significant improvements in those areas as well. Future work can also explore the optimal way to perform this two-dimensional curriculum learning.

The learned policy is small and relatively fast, and can be used on consumer devices without a GPU. In the future, this can be packaged into a lightweight library for neural Clifford synthesis.

A limitation of our current formulation is that one-step action selection becomes expensive for large $n$: the number of possible $\mathrm{CZ}_{i,j}$ actions grows as $O(n^2)$, and the resulting rollouts can require $O(n^2)$ gates. This makes naive inference scale poorly even when the network itself is size-agnostic. Future work could reduce this cost with a factorized action space, hierarchical actions that choose structured multi-gate moves, or policy-guided planning methods such as Monte Carlo tree search during training and inference.

See \href{https://y-richie-y.github.io/clifford/}{\texttt{y-richie-y.github.io/clifford/}} for a live browser demo and recovered Bravyi \textit{et al.} benchmark circuits, including new circuits for previously unrecovered optima.

\section*{Acknowledgements and Disclosures}

We thank Christopher Mingard, Jasmine Brewer, and Ciarán Ryan-Anderson for useful feedback on the manuscript. We also thank the anonymous NeurIPS reviewers for their helpful comments and suggestions. We thank Alexander Koziell-Pipe for interesting discussions and for exploring Monte Carlo tree search for circuit synthesis. RY thanks Dr Simon Harrison for generous support through the Wolfson Harrison UKRI Quantum Foundation Scholarship.
This work is supported by the Engineering and Physical Sciences Research Council grant number EP/Z002230/1, \textit{(De)constructing quantum software (DeQS)}.
The research of RC is supported by the National Research Foundation, Singapore, under its National Research Foundation Fellowship in Artificial Intelligence (Award No.: NRFFIAI1-2024-0014).

\bibliography{refs}

\begin{thebibliography}{48}
\providecommand{\natexlab}[1]{#1}
\providecommand{\url}[1]{\texttt{#1}}
\expandafter\ifx\csname urlstyle\endcsname\relax
  \providecommand{\doi}[1]{doi: #1}\else
  \providecommand{\doi}{doi: \begingroup \urlstyle{rm}\Url}\fi

\bibitem[Aaronson and Gottesman(2004)]{AaronsonGottesman2004}
Scott Aaronson and Daniel Gottesman.
\newblock Improved simulation of stabilizer circuits.
\newblock \emph{Physical Review A}, 70\penalty0 (5), 2004.
\newblock \doi{10.1103/PhysRevA.70.052328}.
\newblock URL \url{https://doi.org/10.1103/physreva.70.052328}.

\bibitem[Agostinelli et~al.(2019)Agostinelli, McAleer, Shmakov, and
  Baldi]{Agostinelli2019DeepCubeA}
Forest Agostinelli, Stephen McAleer, Alexander Shmakov, and Pierre Baldi.
\newblock Solving the rubik’s cube with deep reinforcement learning and
  search.
\newblock \emph{Nature Machine Intelligence}, 1\penalty0 (8):\penalty0
  356--363, 2019.
\newblock \doi{10.1038/s42256-019-0070-z}.
\newblock URL \url{https://doi.org/10.1038/s42256-019-0070-z}.

\bibitem[Amy and Mosca(2019)]{Amy2016TCountReedMuller}
Matthew Amy and Michele Mosca.
\newblock T-count optimization and reed–muller codes.
\newblock \emph{IEEE Transactions on Information Theory}, 65\penalty0
  (8):\penalty0 4771--4784, 2019.
\newblock \doi{10.1109/TIT.2019.2906374}.
\newblock URL \url{https://doi.org/10.1109/tit.2019.2906374}.

\bibitem[Aseguinolaza et~al.(2024)Aseguinolaza, Sobrino, Sobrino,
  Jornet-Somoza, and Borge]{Aseguinolaza2024}
Unai Aseguinolaza, Nahual Sobrino, Gabriel Sobrino, Joaquim Jornet-Somoza, and
  Juan Borge.
\newblock Error estimation in current noisy quantum computers.
\newblock \emph{Quantum Information Processing}, 23\penalty0 (5), 2024.
\newblock \doi{10.1007/s11128-024-04384-z}.
\newblock URL \url{https://doi.org/10.1007/s11128-024-04384-z}.

\bibitem[Bengio et~al.(2009)Bengio, Louradour, Collobert, and
  Weston]{Bengio2009CurriculumLearning}
Yoshua Bengio, Jérôme Louradour, Ronan Collobert, and Jason Weston.
\newblock Curriculum learning.
\newblock In \emph{Proceedings of the 26th Annual International Conference on
  Machine Learning}, pages 41--48, 2009.
\newblock \doi{10.1145/1553374.1553380}.
\newblock URL \url{https://doi.org/10.1145/1553374.1553380}.

\bibitem[Boykin et~al.(2000)Boykin, Mor, Pulver, Roychowdhury, and
  Vatan]{Boykin2000}
P.~Oscar Boykin, Tal Mor, Matthew Pulver, Vwani Roychowdhury, and Farrokh
  Vatan.
\newblock A new universal and fault-tolerant quantum basis.
\newblock \emph{Information Processing Letters}, 75\penalty0 (3):\penalty0
  101--107, 2000.
\newblock \doi{10.1016/S0020-0190(00)00084-3}.
\newblock URL \url{https://doi.org/10.1016/s0020-0190(00)00084-3}.

\bibitem[Bravyi and Kitaev(2005)]{BravyiKitaev2005}
Sergey Bravyi and Alexei Kitaev.
\newblock Universal quantum computation with ideal {Clifford} gates and noisy
  ancillas.
\newblock \emph{Physical Review A}, 71\penalty0 (2), 2005.
\newblock \doi{10.1103/PhysRevA.71.022316}.
\newblock URL \url{https://doi.org/10.1103/physreva.71.022316}.

\bibitem[Bravyi et~al.(2021{\natexlab{a}})Bravyi, Shaydulin, Hu, and
  Maslov]{BravyiTemplates2021}
Sergey Bravyi, Ruslan Shaydulin, Shaohan Hu, and Dmitri Maslov.
\newblock {Clifford} circuit optimization with templates and symbolic {Pauli}
  gates.
\newblock \emph{Quantum}, 5:\penalty0 580, 2021{\natexlab{a}}.
\newblock \doi{10.22331/q-2021-11-16-580}.
\newblock URL \url{https://doi.org/10.22331/q-2021-11-16-580}.

\bibitem[Bravyi et~al.(2021{\natexlab{b}})Bravyi, Shaydulin, Hu, and
  Maslov]{BravyiTemplatesData2021}
Sergey Bravyi, Ruslan Shaydulin, Shaohan Hu, and Dmitri Maslov.
\newblock Data to accompany {Clifford Circuit Optimization with Templates and
  Symbolic Pauli Gates}, 2021{\natexlab{b}}.
\newblock URL
  \url{https://github.com/rsln-s/Clifford_Circuit_Optimization_with_Templates_and_Symbolic_Pauli_Gates}.
\newblock GitHub repository, accessed 2026-05-03.

\bibitem[Bravyi et~al.(2022)Bravyi, Latone, and Maslov]{Bravyi2022}
Sergey Bravyi, Joseph~A. Latone, and Dmitri Maslov.
\newblock 6-qubit optimal {Clifford} circuits.
\newblock \emph{npj Quantum Information}, 8\penalty0 (1), 2022.
\newblock \doi{10.1038/s41534-022-00583-7}.
\newblock URL \url{https://doi.org/10.1038/s41534-022-00583-7}.

\bibitem[Charton et~al.(2023)Charton, Krajenbrink, Meichanetzidis, and
  Yeung]{Yeung2023ZXTransformers}
Francois Charton, Alexandre Krajenbrink, Konstantinos Meichanetzidis, and
  Richie Yeung.
\newblock Teaching small transformers to rewrite {ZX} diagrams.
\newblock In \emph{3rd MATH-AI Workshop at NeurIPS'23}, 2023.
\newblock URL \url{https://mathai2023.github.io/papers/34.pdf}.

\bibitem[Cossio et~al.(2026)Cossio, Bosco, Romanello, Serra, and
  Piazza]{Cossio2026AlphaCNOT}
Jacopo Cossio, Daniele~Lizzio Bosco, Riccardo Romanello, Giuseppe Serra, and
  Carla Piazza.
\newblock {AlphaCNOT}: Learning {CNOT} minimization with model-based planning,
  2026.
\newblock URL \url{https://arxiv.org/abs/2604.13812}.

\bibitem[de~Beaudrap et~al.(2020)de~Beaudrap, Bian, and
  Wang]{DeBeaudrap2020SpiderNest}
Niel de~Beaudrap, Xiaoning Bian, and Quanlong Wang.
\newblock Fast and effective techniques for {T-Count} reduction via spider nest
  identities.
\newblock In \emph{15th Conference on the Theory of Quantum Computation,
  Communication and Cryptography (TQC 2020)}, volume 158, pages 11:1--11:23.
  Schloss Dagstuhl -- Leibniz-Zentrum für Informatik, 2020.
\newblock \doi{10.4230/LIPIcs.TQC.2020.11}.
\newblock URL
  \url{https://drops.dagstuhl.de/entities/document/10.4230/LIPIcs.TQC.2020.11}.

\bibitem[Doherty et~al.(2026)Doherty, Puviani, Brewer, Matos, Amaro, Criger,
  and Stephen]{Doherty2026GraphDecimation}
Michael Doherty, Matteo Puviani, Jasmine Brewer, Gabriel Matos, David Amaro,
  Ben Criger, and David~T. Stephen.
\newblock Fast stabilizer state preparation via {AI}-optimized graph
  decimation, 2026.
\newblock URL \url{https://arxiv.org/abs/2603.17743v1}.

\bibitem[Dubal et~al.(2025)Dubal, Kremer, Martiel, Villar, Wang, and
  Cruz-Benito]{Dubal2025}
Ayushi Dubal, David Kremer, Simon Martiel, Victor Villar, Derek Wang, and Juan
  Cruz-Benito.
\newblock {Pauli} network circuit synthesis with reinforcement learning, 2025.
\newblock URL \url{https://arxiv.org/abs/2503.14448v1}.

\bibitem[Duncan et~al.(2020)Duncan, Kissinger, Perdrix, and van~de
  Wetering]{Duncan2020GraphTheoreticZX}
Ross Duncan, Aleks Kissinger, Simon Perdrix, and John van~de Wetering.
\newblock Graph-theoretic simplification of quantum circuits with the
  {ZX}-calculus.
\newblock \emph{Quantum}, 4:\penalty0 279, 2020.
\newblock \doi{10.22331/q-2020-06-04-279}.
\newblock URL \url{https://doi.org/10.22331/q-2020-06-04-279}.

\bibitem[Fawzi et~al.(2022)Fawzi, Balog, Huang, Hubert, Romera-Paredes,
  Barekatain, Novikov, Ruiz, Schrittwieser, Swirszcz, Silver, Hassabis, and
  Kohli]{Fawzi2022AlphaTensor}
Alhussein Fawzi, Matej Balog, Aja Huang, Thomas Hubert, Bernardino
  Romera-Paredes, Mohammadamin Barekatain, Alexander Novikov, Francisco J.~R.
  Ruiz, Julian Schrittwieser, Grzegorz Swirszcz, David Silver, Demis Hassabis,
  and Pushmeet Kohli.
\newblock Discovering faster matrix multiplication algorithms with
  reinforcement learning.
\newblock \emph{Nature}, 610\penalty0 (7930):\penalty0 47--53, 2022.
\newblock \doi{10.1038/s41586-022-05172-4}.
\newblock URL \url{https://doi.org/10.1038/s41586-022-05172-4}.

\bibitem[Florensa et~al.(2017)Florensa, Held, Wulfmeier, Zhang, and
  Abbeel]{Florensa2017ReverseCurriculum}
Carlos Florensa, David Held, Markus Wulfmeier, Michael Zhang, and Pieter
  Abbeel.
\newblock Reverse curriculum generation for reinforcement learning.
\newblock In \emph{Conference on Robot Learning}, pages 482--495, 2017.
\newblock URL
  \url{http://proceedings.mlr.press/v78/florensa17a/florensa17a.pdf}.

\bibitem[Gidney(2020)]{Gidney2020InverseTableaus}
Craig Gidney.
\newblock Inverting clifford tableaus, 2020.
\newblock URL \url{https://algassert.com/post/2002}.
\newblock Blog post, accessed 2026-04-07.

\bibitem[Gidney et~al.(2024)Gidney, Shutty, and Jones]{Gidney2024Cultivation}
Craig Gidney, Noah Shutty, and Cody Jones.
\newblock Magic state cultivation: growing t states as cheap as {CNOT} gates,
  2024.
\newblock URL \url{https://arxiv.org/abs/2409.17595v1}.

\bibitem[Gottesman(1998)]{Gottesman1998Heisenberg}
Daniel Gottesman.
\newblock The heisenberg representation of quantum computers, 1998.
\newblock URL \url{https://arxiv.org/abs/quant-ph/9807006}.

\bibitem[Heyfron and Campbell(2018)]{Heyfron2017TODD}
Luke~E Heyfron and Earl~T Campbell.
\newblock An efficient quantum compiler that reduces {T} count.
\newblock \emph{Quantum Science and Technology}, 4\penalty0 (1):\penalty0
  015004, 2018.
\newblock \doi{10.1088/2058-9565/aad604}.
\newblock URL \url{https://doi.org/10.1088/2058-9565/aad604}.

\bibitem[{IBM Quantum and Qiskit
  contributors}(2026{\natexlab{a}})]{QiskitCliffordGreedyDocs}
{IBM Quantum and Qiskit contributors}.
\newblock {GreedySynthesisClifford}, 2026{\natexlab{a}}.
\newblock URL
  \url{https://quantum.cloud.ibm.com/docs/en/api/qiskit/2.2/qiskit.transpiler.passes.synthesis.hls_plugins.GreedySynthesisClifford}.
\newblock Qiskit 2.2 API documentation, accessed 2026-05-07.

\bibitem[{IBM Quantum and Qiskit
  contributors}(2026{\natexlab{b}})]{QiskitCliffordPluginDocs}
{IBM Quantum and Qiskit contributors}.
\newblock {DefaultSynthesisClifford}, 2026{\natexlab{b}}.
\newblock URL
  \url{https://quantum.cloud.ibm.com/docs/en/api/qiskit/2.2/qiskit.transpiler.passes.synthesis.hls_plugins.DefaultSynthesisClifford}.
\newblock Qiskit 2.2 API documentation, accessed 2026-05-07.

\bibitem[{IBM Quantum and Qiskit
  contributors}(2026{\natexlab{c}})]{QiskitRandomCliffordDocs}
{IBM Quantum and Qiskit contributors}.
\newblock {random\_clifford}, 2026{\natexlab{c}}.
\newblock URL
  \url{https://quantum.cloud.ibm.com/docs/en/api/qiskit/2.2/quantum_info}.
\newblock Qiskit 2.2 API documentation, accessed 2026-05-07.

\bibitem[Kissinger and van~de Wetering(2020)]{Kissinger2020ReducingNonClifford}
Aleks Kissinger and John van~de Wetering.
\newblock Reducing the number of non-{Clifford} gates in quantum circuits.
\newblock \emph{Physical Review A}, 102\penalty0 (2), 2020.
\newblock \doi{10.1103/PhysRevA.102.022406}.
\newblock URL \url{https://doi.org/10.1103/physreva.102.022406}.

\bibitem[Kremer et~al.(2024)Kremer, Villar, Paik, Duran, Faro, and
  Cruz-Benito]{kremer2024practical}
David Kremer, Victor Villar, Hanhee Paik, Ivan Duran, Ismael Faro, and Juan
  Cruz-Benito.
\newblock Practical and efficient quantum circuit synthesis and transpiling
  with reinforcement learning, 2024.
\newblock URL \url{https://arxiv.org/abs/2405.13196v2}.

\bibitem[Litinski(2019)]{Litinski2019MagicStateDistillation}
Daniel Litinski.
\newblock Magic state distillation: Not as costly as you think.
\newblock \emph{Quantum}, 3:\penalty0 205, 2019.
\newblock \doi{10.22331/q-2019-12-02-205}.
\newblock URL \url{https://doi.org/10.22331/q-2019-12-02-205}.

\bibitem[Litinski and von Oppen(2018)]{LitinskiOppen2018}
Daniel Litinski and Felix von Oppen.
\newblock Lattice surgery with a twist: Simplifying {Clifford} gates of surface
  codes.
\newblock \emph{Quantum}, 2:\penalty0 62, 2018.
\newblock \doi{10.22331/q-2018-05-04-62}.
\newblock URL \url{https://doi.org/10.22331/q-2018-05-04-62}.

\bibitem[Mattick et~al.(2025)Mattick, Periyasamy, Ufrecht, Dubey, Mutschler,
  Plinge, and Scherer]{Mattick2025ZXRL}
Alexander Mattick, Maniraman Periyasamy, Christian Ufrecht, Abhishek~Y. Dubey,
  Christopher Mutschler, Axel Plinge, and Daniel~D. Scherer.
\newblock Optimizing quantum circuits via {ZX} diagrams using reinforcement
  learning and graph neural networks, 2025.
\newblock URL \url{https://arxiv.org/abs/2504.03429v1}.

\bibitem[Narvekar et~al.(2020)Narvekar, Peng, Leonetti, Sinapov, Taylor, and
  Stone]{Narvekar2020CurriculumSurvey}
Sanmit Narvekar, Bei Peng, Matteo Leonetti, Jivko Sinapov, Matthew~E. Taylor,
  and Peter Stone.
\newblock Curriculum learning for reinforcement learning domains: A framework
  and survey.
\newblock \emph{Journal of Machine Learning Research}, 21\penalty0
  (181):\penalty0 1--50, 2020.
\newblock URL \url{http://jmlr.org/papers/volume21/20-212/20-212.pdf}.

\bibitem[Nägele and Marquardt(2024)]{Nagele2024ZXRL}
Maximilian Nägele and Florian Marquardt.
\newblock Optimizing {ZX}-diagrams with deep reinforcement learning.
\newblock \emph{Machine Learning: Science and Technology}, 5\penalty0
  (3):\penalty0 035077, 2024.
\newblock \doi{10.1088/2632-2153/ad76f7}.
\newblock URL \url{https://doi.org/10.1088/2632-2153/ad76f7}.

\bibitem[Peham et~al.(2023)Peham, Brandl, Kueng, Wille, and
  Burgholzer]{Peham2023}
Tom Peham, Nina Brandl, Richard Kueng, Robert Wille, and Lukas Burgholzer.
\newblock Depth-optimal synthesis of {Clifford} circuits with {SAT} solvers.
\newblock In \emph{2023 IEEE International Conference on Quantum Computing and
  Engineering (QCE)}, pages 802--813, 2023.
\newblock \doi{10.1109/QCE57702.2023.00095}.
\newblock URL \url{https://doi.org/10.1109/qce57702.2023.00095}.

\bibitem[Riu et~al.(2025)Riu, Nogué, Vilaplana, Garcia-Saez, and
  Estarellas]{Riu2025ZXRL}
Jordi Riu, Jan Nogué, Gerard Vilaplana, Artur Garcia-Saez, and Marta~P.
  Estarellas.
\newblock Reinforcement learning based quantum circuit optimization via
  {ZX}-calculus.
\newblock \emph{Quantum}, 9:\penalty0 1758, 2025.
\newblock \doi{10.22331/Q-2025-05-28-1758}.
\newblock URL \url{https://doi.org/10.22331/q-2025-05-28-1758}.

\bibitem[Ruiz et~al.(2025)Ruiz, Laakkonen, Bausch, Balog, Barekatain, Heras,
  Novikov, Fitzpatrick, Romera-Paredes, van~de Wetering, Fawzi, Meichanetzidis,
  and Kohli]{Ruiz2025AlphaTensorQuantum}
Francisco J.~R. Ruiz, Tuomas Laakkonen, Johannes Bausch, Matej Balog,
  Mohammadamin Barekatain, Francisco J.~H. Heras, Alexander Novikov, Nathan
  Fitzpatrick, Bernardino Romera-Paredes, John van~de Wetering, Alhussein
  Fawzi, Konstantinos Meichanetzidis, and Pushmeet Kohli.
\newblock Quantum circuit optimization with {AlphaTensor}.
\newblock \emph{Nature Machine Intelligence}, 7\penalty0 (3):\penalty0
  374--385, 2025.
\newblock \doi{10.1038/s42256-025-01001-1}.
\newblock URL \url{https://doi.org/10.1038/s42256-025-01001-1}.

\bibitem[Schulman et~al.(2017)Schulman, Wolski, Dhariwal, Radford, and
  Klimov]{Schulman2017PPO}
John Schulman, Filip Wolski, Prafulla Dhariwal, Alec Radford, and Oleg Klimov.
\newblock Proximal policy optimization algorithms, 2017.
\newblock URL \url{https://arxiv.org/abs/1707.06347v2}.

\bibitem[Selinger(2015)]{Selinger2015}
Peter Selinger.
\newblock Efficient {Clifford+T} approximation of single-qubit operators.
\newblock \emph{Quantum Information and Computation}, 15\penalty0
  (1\&2):\penalty0 159--180, 2015.
\newblock \doi{10.26421/qic15.1-2-10}.
\newblock URL \url{https://doi.org/10.26421/qic15.1-2-10}.

\bibitem[Shaik and van~de Pol(2025)]{Shaik2025}
Irfansha Shaik and Jaco van~de Pol.
\newblock {CNOT}-optimal {Clifford} synthesis as {SAT}.
\newblock In \emph{28th International Conference on Theory and Applications of
  Satisfiability Testing (SAT 2025)}, volume 341, pages 28:1--28:21. Schloss
  Dagstuhl -- Leibniz-Zentrum für Informatik, 2025.
\newblock \doi{10.4230/LIPIcs.SAT.2025.28}.
\newblock URL
  \url{https://drops.dagstuhl.de/entities/document/10.4230/LIPIcs.SAT.2025.28}.

\bibitem[Shaw et~al.(2018)Shaw, Uszkoreit, and
  Vaswani]{Shaw2018RelativeAttention}
Peter Shaw, Jakob Uszkoreit, and Ashish Vaswani.
\newblock Self-attention with relative position representations.
\newblock In \emph{Proceedings of the 2018 Conference of the North American
  Chapter of the Association for Computational Linguistics: Human Language
  Technologies, Volume 2 (Short Papers)}, pages 464--468, New Orleans,
  Louisiana, 2018. Association for Computational Linguistics.
\newblock \doi{10.18653/v1/N18-2074}.
\newblock URL \url{https://aclanthology.org/N18-2074/}.

\bibitem[Simmons(2021)]{Simmons2021PauliFlow}
Will Simmons.
\newblock Relating measurement patterns to circuits via {Pauli} flow.
\newblock \emph{Electronic Proceedings in Theoretical Computer Science},
  343:\penalty0 50--101, 2021.
\newblock \doi{10.4204/EPTCS.343.4}.
\newblock URL \url{https://doi.org/10.4204/eptcs.343.4}.

\bibitem[Suarez(2024)]{Suarez2024PufferLib}
Joseph Suarez.
\newblock {PufferLib}: Making reinforcement learning libraries and environments
  play nice, 2024.
\newblock URL \url{https://arxiv.org/abs/2406.12905v1}.

\bibitem[van~de Griend(2019)]{meijermasters}
Arianne van~de Griend.
\newblock Constrained quantum {CNOT} circuit re-synthesis using deep
  reinforcement learning.
\newblock Master's thesis, Radboud University, 2019.
\newblock URL \url{https://theses.ubn.ru.nl/handle/123456789/10713}.

\bibitem[van~de Wetering et~al.(2025)van~de Wetering, Yeung, Laakkonen, and
  Kissinger]{Wetering2025ParametrisedCompilation}
John van~de Wetering, Richie Yeung, Tuomas Laakkonen, and Aleks Kissinger.
\newblock Optimal compilation of parametrised quantum circuits.
\newblock \emph{Quantum}, 9:\penalty0 1828, 2025.
\newblock \doi{10.22331/q-2025-08-27-1828}.
\newblock URL \url{https://doi.org/10.22331/q-2025-08-27-1828}.

\bibitem[van~der Pol et~al.(2020)van~der Pol, Worrall, van Hoof, Oliehoek, and
  Welling]{VanDerPol2020MDPHomomorphicNetworks}
Elise van~der Pol, Daniel~E. Worrall, Herke van Hoof, Frans~A. Oliehoek, and
  Max Welling.
\newblock {MDP} homomorphic networks: Group symmetries in reinforcement
  learning.
\newblock In \emph{Advances in Neural Information Processing Systems},
  volume~33, pages 4199--4210, 2020.
\newblock URL
  \url{https://proceedings.neurips.cc/paper/2020/file/2be5f9c2e3620eb73c2972d7552b6cb5-Paper.pdf}.

\bibitem[Vaswani et~al.(2017)Vaswani, Shazeer, Parmar, Uszkoreit, Jones, Gomez,
  Kaiser, and Polosukhin]{Vaswani2017Attention}
Ashish Vaswani, Noam Shazeer, Niki Parmar, Jakob Uszkoreit, Llion Jones,
  Aidan~N. Gomez, Lukasz Kaiser, and Illia Polosukhin.
\newblock Attention is all you need.
\newblock In \emph{Advances in Neural Information Processing Systems},
  volume~30, 2017.
\newblock URL
  \url{https://proceedings.neurips.cc/paper/2017/hash/3f5ee243547dee91fbd053c1c4a845aa-Abstract.html}.

\bibitem[Webster et~al.(2026)Webster, Koutsioumpas, and
  Browne]{Webster2025HeuristicOptimal}
Mark Webster, Stergios Koutsioumpas, and Dan~E Browne.
\newblock Heuristic and optimal synthesis of {CNOT} and {Clifford} circuits.
\newblock \emph{Quantum}, 10:\penalty0 2212, 2026.
\newblock \doi{10.22331/q-2026-09-21-2212}.
\newblock URL \url{https://doi.org/10.22331/q-2026-09-21-2212}.

\bibitem[Zen et~al.(2025)Zen, Olle, Colmenarez, Puviani, Müller, and
  Marquardt]{Zen2023}
Remmy Zen, Jan Olle, Luis Colmenarez, Matteo Puviani, Markus Müller, and
  Florian Marquardt.
\newblock Quantum circuit discovery for fault-tolerant logical state
  preparation with reinforcement learning.
\newblock \emph{Physical Review X}, 15\penalty0 (4):\penalty0 041012, 2025.
\newblock \doi{10.1103/gqpr-dgz7}.
\newblock URL \url{https://doi.org/10.1103/gqpr-dgz7}.

\bibitem[Zinkevich and Balch(2001)]{ZinkevichBalch2001SymmetryMDP}
Martin Zinkevich and Tucker~R. Balch.
\newblock Symmetry in {Markov} decision processes and its implications for
  single agent and multiagent learning.
\newblock In \emph{Proceedings of the Eighteenth International Conference on
  Machine Learning}, pages 632--640. Morgan Kaufmann, 2001.
\newblock URL \url{https://www.cs.cmu.edu/~maz/publications/symmetry7.pdf}.

\end{thebibliography}
\bibliographystyle{plainnat}

\clearpage
\appendix
\setcounter{table}{0}
\renewcommand{\thetable}{A\arabic{table}}
\renewcommand{\theHtable}{appendix.table.\arabic{table}}
\setcounter{figure}{0}
\renewcommand{\thefigure}{A\arabic{figure}}
\renewcommand{\theHfigure}{appendix.figure.\arabic{figure}}

\section{Equivariant Optimal Policies}
\label{app:equivariant-optimal-policy}

This appendix explains why we can use an equivariant policy class in the Clifford synthesis MDP, following the standard symmetry theory of finite MDPs \cite{VanDerPol2020MDPHomomorphicNetworks, ZinkevichBalch2001SymmetryMDP}. The argument applies more generally to stationary discounted MDPs whose dynamics and rewards commute with a group action.

\paragraph{Proposition.}
Let $(\mathcal{S},\mathcal{A},\mathcal{T},\mathcal{R})$ be a finite deterministic MDP with stationary transition and reward functions and discount factor $\gamma\in[0,1)$. Let a group $\Gamma$ act on both $\mathcal{S}$ and $\mathcal{A}$. Suppose that, for every $\sigma\in\Gamma$, $M\in\mathcal{S}$, and $G\in\mathcal{A}$,
\[
    \mathcal{T}(\sigma\!\cdot\! M,\sigma\!\cdot\! G)
    =
    \sigma\!\cdot\!\mathcal{T}(M,G),
    \qquad
    \mathcal{R}(\sigma\!\cdot\! M,\sigma\!\cdot\! G)
    =
    \mathcal{R}(M,G).
\]
Then there exists an optimal stationary stochastic Markov policy $\pi^\star$ that is equivariant:
\[
    \pi^\star(\sigma\!\cdot\! G \mid \sigma\!\cdot\! M)
    =
    \pi^\star(G\mid M)
\]
for all states $M$, actions $G$, and group elements $\sigma$.

\paragraph{Proof.}
Let $V^\star$ be the unique solution to the discounted Bellman optimality equation
\[
    V^\star(M)
    =
    \max_{G\in\mathcal{A}}
    \left[
        \mathcal{R}(M,G)
        +
        \gamma V^\star(\mathcal{T}(M,G))
    \right].
\]
For any fixed $\sigma\in\Gamma$, define $W_\sigma(M)=V^\star(\sigma\!\cdot\! M)$. Then
\[
\begin{aligned}
W_\sigma(M)
&=
V^\star(\sigma\!\cdot\! M)\\
&=
\max_{G'\in\mathcal{A}}
\left[
    \mathcal{R}(\sigma\!\cdot\! M,G')
    +
    \gamma V^\star(\mathcal{T}(\sigma\!\cdot\! M,G'))
\right]\\
&=
\max_{G\in\mathcal{A}}
\left[
    \mathcal{R}(\sigma\!\cdot\! M,\sigma\!\cdot\! G)
    +
    \gamma V^\star(\mathcal{T}(\sigma\!\cdot\! M,\sigma\!\cdot\! G))
\right]\\
&=
\max_{G\in\mathcal{A}}
\left[
    \mathcal{R}(M,G)
    +
    \gamma V^\star(\sigma\!\cdot\!\mathcal{T}(M,G))
\right]\\
&=
\max_{G\in\mathcal{A}}
\left[
    \mathcal{R}(M,G)
    +
    \gamma W_\sigma(\mathcal{T}(M,G))
\right].
\end{aligned}
\]
Here we reindexed the maximisation using the bijection $G\mapsto\sigma\!\cdot\! G$. Thus $W_\sigma$ also satisfies the Bellman optimality equation. By uniqueness, $W_\sigma=V^\star$, so
\[
    V^\star(\sigma\!\cdot\! M)=V^\star(M).
\]
Define the optimal action-value function
\[
    Q^\star(M,G)
    =
    \mathcal{R}(M,G)
    +
    \gamma V^\star(\mathcal{T}(M,G)).
\]
The invariance of $V^\star$, together with the equivariance of $\mathcal{T}$ and invariance of $\mathcal{R}$, gives
\[
    Q^\star(\sigma\!\cdot\! M,\sigma\!\cdot\! G)=Q^\star(M,G).
\]
Therefore the optimal action set
\[
    \mathcal{A}^\star(M)
    =
    \arg\max_{G\in\mathcal{A}} Q^\star(M,G)
\]
satisfies
\[
    \mathcal{A}^\star(\sigma\!\cdot\! M)
    =
    \sigma\!\cdot\!\mathcal{A}^\star(M).
\]
Let $\pi^\star(\cdot\mid M)$ be the uniform distribution over $\mathcal{A}^\star(M)$. This policy is optimal because it assigns probability only to optimal actions. It is also equivariant, since the optimal action set at $\sigma\!\cdot\! M$ is exactly the relabeling of the optimal action set at $M$. Hence
\[
    \pi^\star(\sigma\!\cdot\! G\mid \sigma\!\cdot\! M)
    =
    \pi^\star(G\mid M).
\]
\hfill$\square$

\paragraph{Application to Clifford synthesis.}
In our Clifford synthesis MDP, the symmetry group is the qubit permutation group $S_n$. A permutation $\sigma\in S_n$ acts on tableaus by simultaneously relabeling the row and column qubit blocks, and acts on generators by
\[
H_i\mapsto H_{\sigma(i)},\qquad
S_i\mapsto S_{\sigma(i)},\qquad
\mathrm{CZ}_{i,j}\mapsto \mathrm{CZ}_{\sigma(i),\sigma(j)}.
\]
Right multiplication by generators commutes with this relabeling, so
\[
    \mathcal{T}(\sigma\!\cdot\! M,\sigma\!\cdot\! G)
    =
    \sigma\!\cdot\!\mathcal{T}(M,G).
\]
The reward is also invariant: the single-qubit and two-qubit costs depend only on the gate type, the identity tableau is fixed by every qubit relabeling, and the Hamming distance to the identity is unchanged by simultaneously permuting rows and columns. Therefore the proposition applies, and there exists an optimal stationary stochastic Markov policy for the discounted Clifford synthesis objective that is equivariant under qubit relabeling. If a hard rollout cap is included as part of the objective, we can augment the state with the remaining step budget and let $S_n$ act trivially on that extra coordinate. The same argument then gives an optimal equivariant policy on the augmented state space.

\section{Random-Walk Limit Distribution}
\label{app:random-walk-limit}

This appendix justifies the use of uniformly sampled Clifford targets as the $\infty$ endpoint of the random-walk curriculum.

\begin{theorem}
\label{thm:random-walk-limit}
Fix $n\geq 3$, write $\Gamma=\Sp(2n,\mathbb{F}_2)$, and let $K$ be the transition matrix for the random walk on $\Gamma$ that right-multiplies by a uniformly sampled generator from $\mathcal{G}$ in \eqref{eq:clifford-generators}.
Then the $d$-step distribution from the identity converges to the uniform distribution on $\Gamma$:
\[
    K^d(I_{2n},M)\longrightarrow \frac{1}{|\Gamma|}
    \qquad
    \text{for every }M\in\Gamma.
\]
For $n=1$ and $n=2$, the corresponding walk has period $2$ and therefore does not converge to the uniform distribution at fixed walk length.
\end{theorem}

\begin{proof}
We use the standard ergodic theorem for finite Markov chains: an irreducible and aperiodic finite chain converges to its unique stationary distribution.
Because $\mathcal{G}$ generates $\Gamma$, the associated Cayley graph is connected, so the chain is irreducible.
Figure~\ref{fig:appendix-cayley-graph} shows a local ball in the corresponding generator graph for the two-qubit case, where random-walk steps move along generator-labeled edges from the identity.
The uniform distribution is stationary: if $M$ is uniform on $\Gamma$ and $G\in\mathcal{G}$ is fixed, then $MG$ is still uniform, because right multiplication by $G$ only permutes the elements of $\Gamma$.

For convergence to this stationary distribution we also need aperiodicity.
Each generator is an involution in the binary symplectic representation, so $G^2=I_{2n}$ gives a closed walk of length $2$.
For $n\geq 3$, there is also an odd closed walk:
\[
H_1\,\mathrm{CZ}_{1,2}\,H_1\,\mathrm{CZ}_{1,3}\,H_1\,
\mathrm{CZ}_{2,3}\,\mathrm{CZ}_{1,2}\,H_1\,\mathrm{CZ}_{1,3}
=I_{2n}.
\]
Figure~\ref{fig:odd-walk-identity} draws this nine-step identity circuit.
Equivalently, the identity is verified by multiplying the generator matrices defined in Appendix~\ref{app:clifford-tableaus}; it acts nontrivially only on qubits $1,2,3$, so the same relation embeds in every larger $n$.
Thus the period of the identity divides both $2$ and $9$, and is therefore $1$.
Irreducibility implies every state has the same period, so the whole chain is finite, irreducible, and aperiodic.
The standard convergence theorem for finite Markov chains then gives
\[
    K^d(I_{2n},M)\longrightarrow \frac{1}{|\Gamma|}
    \qquad
    \text{for every }M\in\Gamma,
\]
which is exactly uniform sampling from $\Sp(2n,\mathbb{F}_2)$.

For $n=1$ and $n=2$, the walk generated by $\{H_i,S_i,\mathrm{CZ}_{i,j}\}$ has period $2$, so the fixed-length distributions oscillate between the two parity classes of the Cayley graph rather than converging.
Concretely, for $n=1$ the generators are only $H_1$ and $S_1$, and the two parity classes are
\[
    \{I,H_1S_1,S_1H_1\}
    \qquad\text{and}\qquad
    \{H_1,S_1,H_1S_1H_1\},
\]
the elements reachable at even and odd times, respectively.
For $n=2$, $Sp(4,\mathbb{F}_2)$ has $720$ elements, split into two parity classes of size $360$ under this generator walk.
\end{proof}

\begin{figure}[t]
\centering
\includegraphics[width=0.88\linewidth]{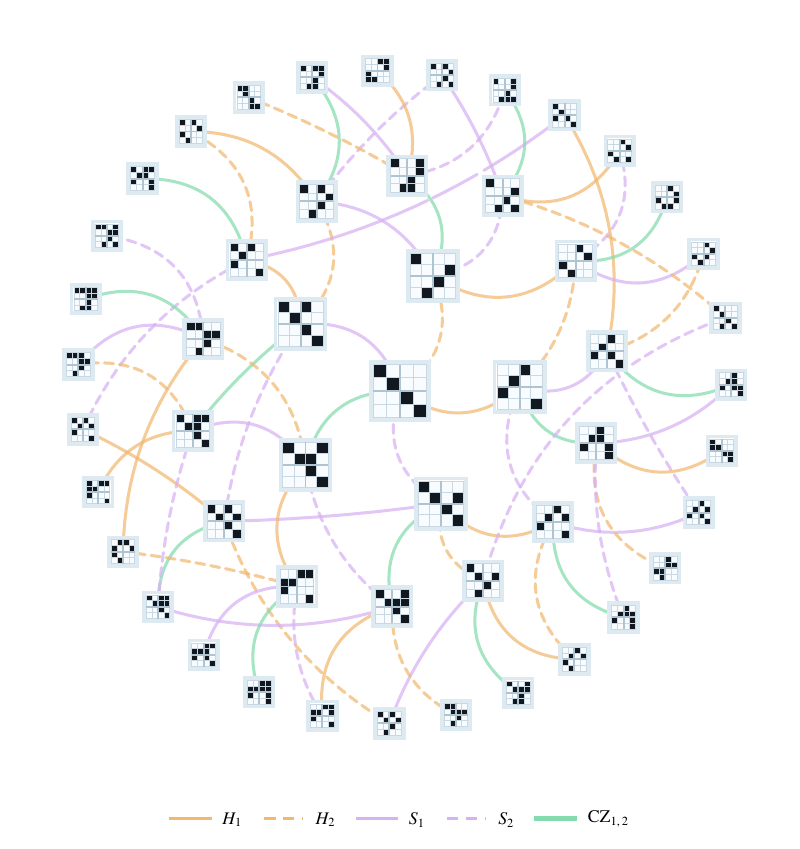}
\caption{Local view of the Cayley graph induced by the two-qubit Clifford generators. Vertices are tableaus reached from the identity by walks of length at most three; the legend identifies the generator for each edge.}
\label{fig:appendix-cayley-graph}
\end{figure}

\begin{figure}[t]
\centering
\includegraphics[width=0.78\linewidth]{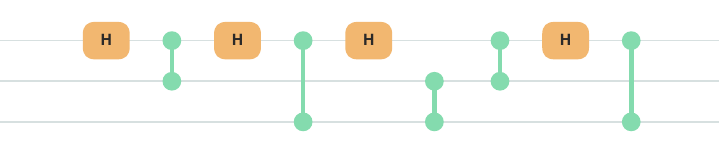}
\caption{Nine-gate identity circuit used in Theorem~\ref{thm:random-walk-limit}.}
\label{fig:odd-walk-identity}
\end{figure}

\section{Architecture Implementation Details}
\label{app:architecture-details}

This appendix gives the implementation details behind the embedding, aggregation, message-passing, and readout stages described in Section~\ref{sec:method}. All additional features below are permutation-equivariant count-based features or permutation-invariant global summaries, so including them does not change the symmetry argument in the main text.

\paragraph{Embedding.}
For each $2\times2$ block $M_{ij}$ of the reshaped tableau, the embedding lookup uses the four binary entries of $M_{ij}$ together with a diagonal indicator $\ind(i=j)$, giving $32$ possible block states.
This produces the edge embedding $e_{ij}\in\mathbb{R}^h$ used in the main text.

\paragraph{Aggregation and count features.}
The initial qubit token $q_i^{(0)}$ is built from row-pooled and column-pooled edge embeddings, the diagonal edge embedding $e_{ii}$, and a count-feature vector $\eta_i$.
Specifically, with a shared MLP $\phi_{\mathrm{node}}$, we set
\begin{equation} \label{eq:initial-node-embedding}
    q_i^{(0)} \coloneqq \phi_{\mathrm{node}}(\underbrace{\left[\operatorname{mean}_j e_{ij}, \operatorname{max}_j e_{ij}\right]}_{\text{Row features}}, \underbrace{\left[\operatorname{mean}_j e_{ji}, \operatorname{max}_j e_{ji}\right]}_{\text{Column features}}, \underbrace{e_{ii}}_{\text{Diagonal embedding}}, \eta_i).
\end{equation}
Here $\eta_i \coloneqq [\eta_{i,1}, \eta_{i,2}] \in \mathbb{R}^{11}$.
The choice of these features is motivated by Webster \emph{et al.} ~\cite{Webster2025HeuristicOptimal}, who show a strong correlation between related rank-count features and the optimal two-qubit gate count for stabilizer state preparation.
The first component $\eta_{i,1}\in\mathbb{R}^9$ is
\[
\eta_{i,1}=
\bigl[
d_i,\;
\rho^{\mathrm{row}}_i,\rho^{\mathrm{col}}_i,\;
\rho^{\mathrm{row},1}_i,\rho^{\mathrm{row},2}_i,\;
\rho^{\mathrm{col},1}_i,\rho^{\mathrm{col},2}_i,\;
\rho^{\mathrm{offrow}}_i,\rho^{\mathrm{offcol}}_i
\bigr],
\]
where $d_i$ indicates whether the diagonal block $M_{ii}$ is the $2\times2$ identity; $\rho^{\mathrm{row}}_i$ and $\rho^{\mathrm{col}}_i$ are the fractions of nonzero blocks in block row $i$ and block column $i$; $\rho^{\mathrm{row},1}_i,\rho^{\mathrm{row},2}_i$ and $\rho^{\mathrm{col},1}_i,\rho^{\mathrm{col},2}_i$ are the corresponding fractions of rank-one and rank-two blocks; and $\rho^{\mathrm{offrow}}_i,\rho^{\mathrm{offcol}}_i$ are the fractions of nonzero off-diagonal blocks in row $i$ and column $i$, normalized by $n-1$.
The vector $\eta_{i,2}\in\mathbb{R}^2$ records the rank of the diagonal block $M_{ii}$ as two indicators for rank one and rank two.

\paragraph{Message passing.}
The message function $\phi_{\mathrm{msg}}$ and update function $\phi_{\mathrm{upd}}$ are shared across all qubit indices and all message-passing rounds.
As in the main text, the update for qubit $i$ uses the current token $q_i^{(k)}$, mean-pooled incoming messages, max-pooled incoming messages, and the rank-count component $\eta_{i,1}$.

\paragraph{Readout.}
Let $q_i=q_i^{(L)}$ be the final qubit token after $L$ message-passing rounds.
The implementation first forms invariant statistics
\begin{align*}
    s &\coloneqq
    \left[
    \operatorname{mean}_i q_i,\;
    \operatorname{max}_i q_i,\;
    \operatorname{std}_i q_i,\;
    \bar c
    \right],\\
    \bar c &\coloneqq
    \left[
    \operatorname{mean}_i d_i,\;
    \operatorname{mean}_i \rho^{\mathrm{row}}_i,\;
    \operatorname{mean}_i \rho^{\mathrm{col}}_i,\;
    \operatorname{mean}_i \rho^{\mathrm{offrow}}_i,\;
    \operatorname{mean}_i \rho^{\mathrm{offcol}}_i,\;
    \operatorname{mean}_i \eta_{i,2}
    \right],
\end{align*}
where $\operatorname{std}_i$ is the elementwise standard deviation over qubit tokens.
Thus the invariant readout statistics have dimension $3h+7$.
The implementation maps these statistics through a global MLP to obtain the learned summary $g$.

For the one-qubit action heads, the readout can be written as the separate MLPs $\phi_{\mathrm{H}}$ and $\phi_{\mathrm{S}}$ used in the main text.
Concretely, the implementation uses a local MLP conditioned on a one-hot vector encoding the action type:
\[
    (q_i,\;e_{ii},\;g,\;\eta_i,\;\tau_a),
\]
where $\tau_a\in\mathbb{R}^2$ encodes the action type.
The two-qubit $\mathrm{CZ}$ head receives $(p_{ij},g)$, where the symmetric pair features are
\[
    p_{ij}\coloneqq [q_i+q_j,\; q_i\odot q_j,\; |q_i-q_j|,\; e_{ij}+e_{ji},\; e_{ij}\odot e_{ji}].
\]
Here $\odot$ denotes elementwise multiplication.
The value head is the MLP $\phi_V(s)$ applied directly to the invariant statistics $s$.

\section{Clifford Circuits and Tableaus}
\label{app:clifford-tableaus}

\begin{figure*}[t]
\centering
\begin{minipage}[c]{0.85\textwidth}
\centering
\includegraphics[width=\textwidth]{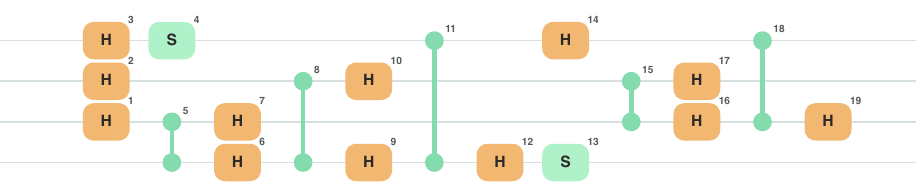}
\end{minipage}

\vspace{2em}

\begin{minipage}[c]{0.85\textwidth}
\centering
\newcommand{\rolloutstate}[1]{\raisebox{-0.5\height}{\includegraphics[width=0.16\textwidth]{#1}}}
\newcommand{\rolloutarrow}[1]{\makebox[0.052\textwidth][c]{$\overset{#1}{\Rightarrow}$}}
\newcommand{\rolloutblankarrow}{\makebox[0.052\textwidth][c]{}}
\newcommand{\rolloutphantomarrow}{\makebox[0.052\textwidth][c]{\phantom{$\overset{\mathrm{CZ}_{1,3}}{\Rightarrow}$}}}
\noindent\rolloutblankarrow\rolloutphantomarrow\hspace{0.00em}
\rolloutstate{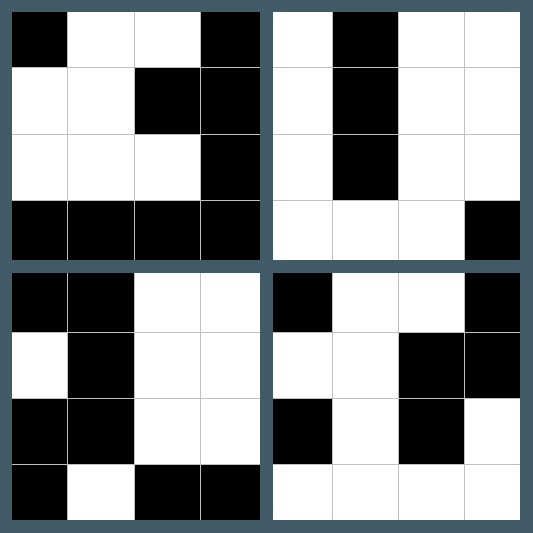}
\rolloutarrow{H_3}
\rolloutstate{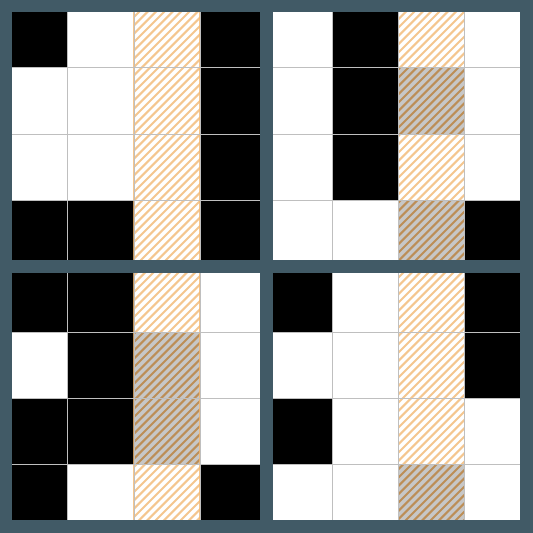}
\rolloutarrow{H_2}
\rolloutstate{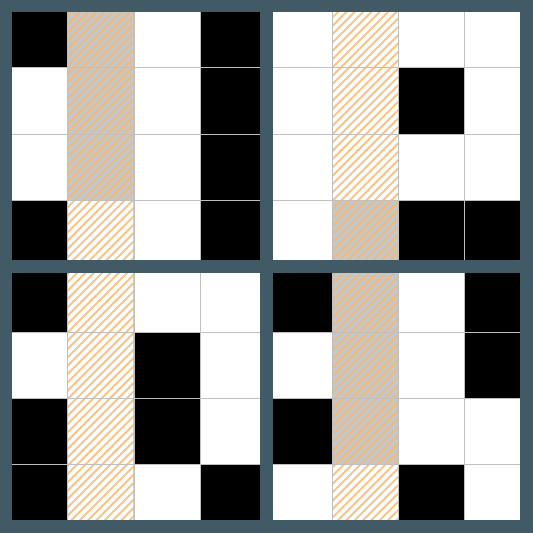}
\rolloutarrow{H_1}
\rolloutstate{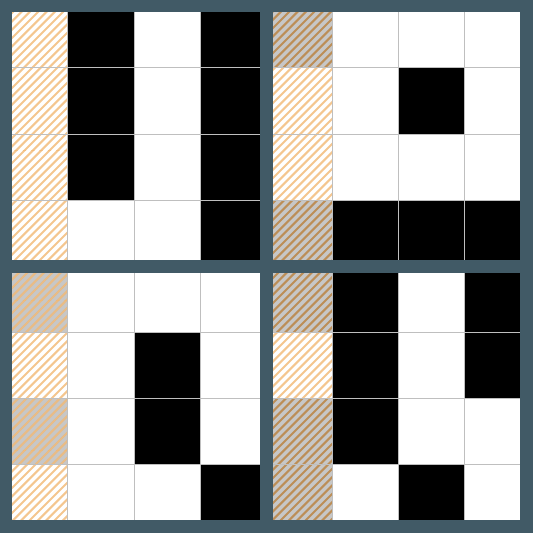}

\medskip

\noindent\rolloutblankarrow
\rolloutarrow{S_1}
\rolloutstate{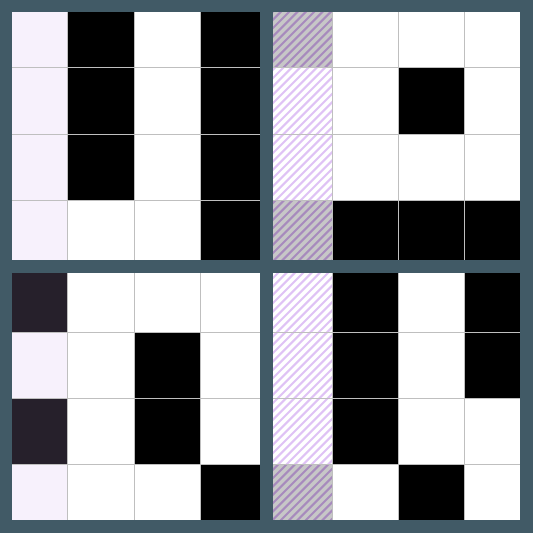}
\rolloutarrow{\mathrm{CZ}_{3,4}}
\rolloutstate{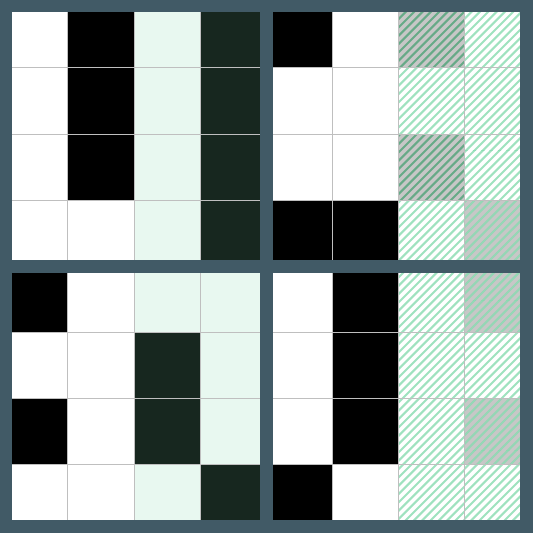}
\rolloutarrow{H_4}
\rolloutstate{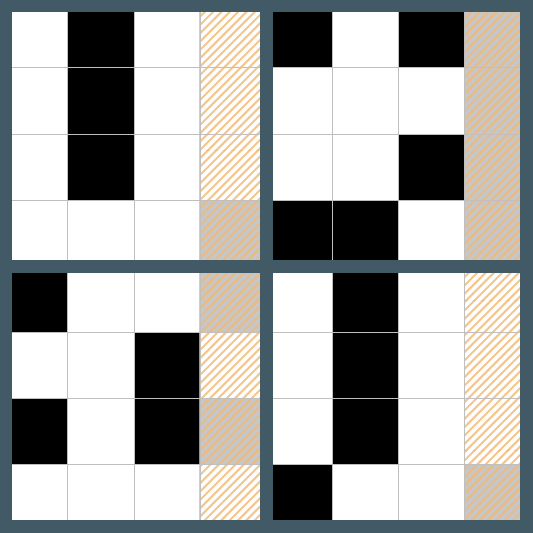}
\rolloutarrow{H_3}
\rolloutstate{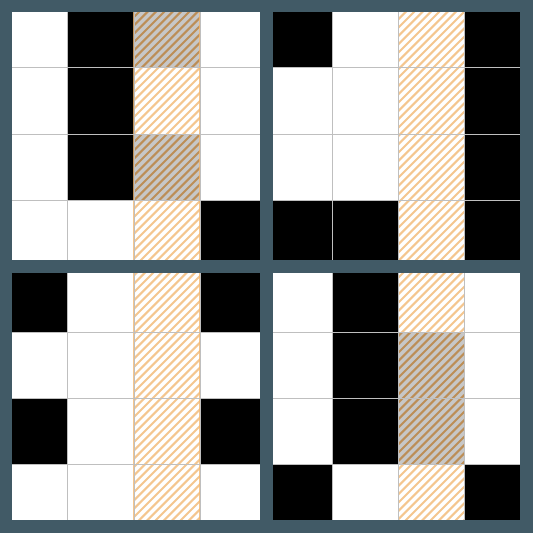}

\medskip

\noindent\rolloutblankarrow
\rolloutarrow{\mathrm{CZ}_{2,4}}
\rolloutstate{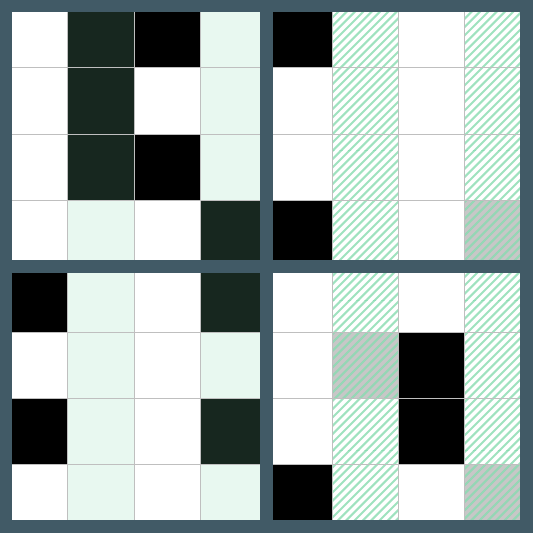}
\rolloutarrow{H_4}
\rolloutstate{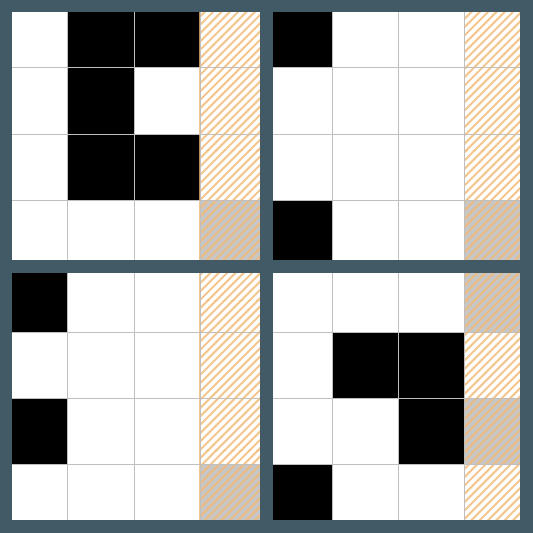}
\rolloutarrow{H_2}
\rolloutstate{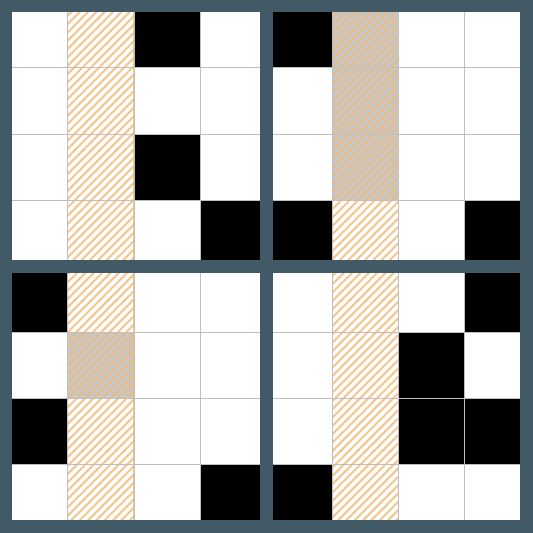}
\rolloutarrow{\mathrm{CZ}_{1,4}}
\rolloutstate{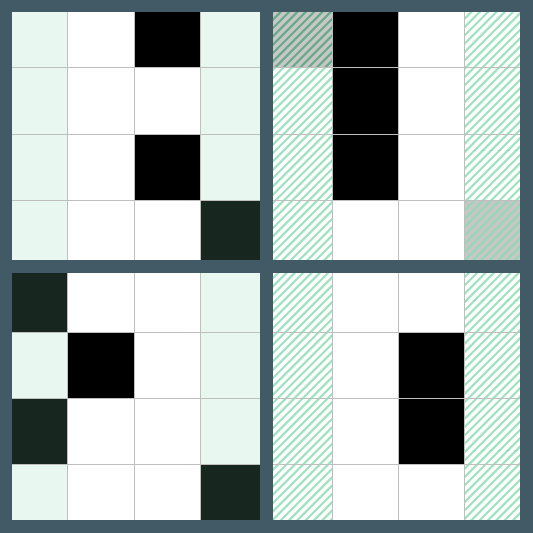}

\medskip

\noindent\rolloutblankarrow
\rolloutarrow{H_4}
\rolloutstate{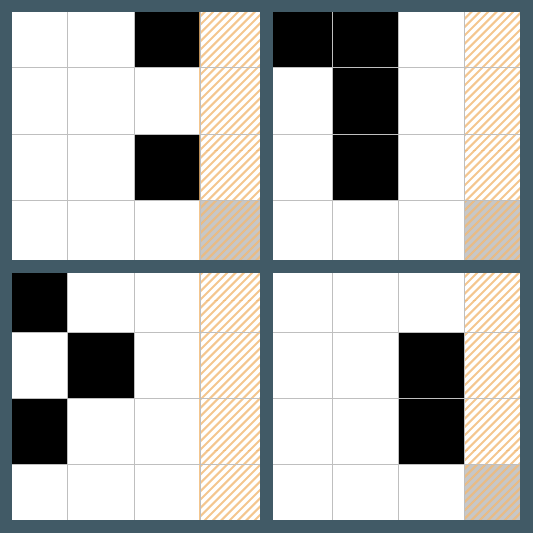}
\rolloutarrow{S_4}
\rolloutstate{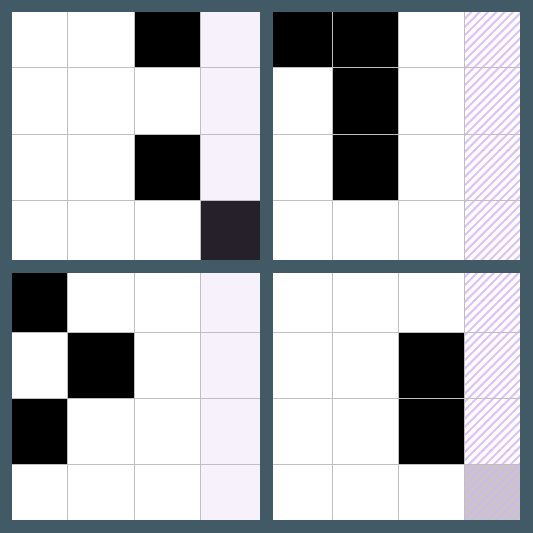}
\rolloutarrow{H_1}
\rolloutstate{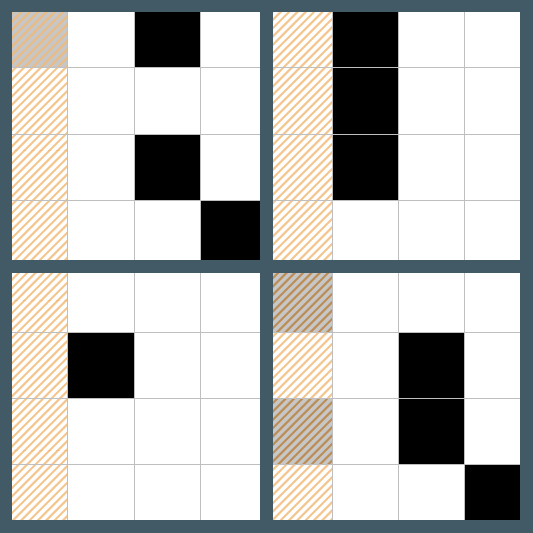}
\rolloutarrow{\mathrm{CZ}_{2,3}}
\rolloutstate{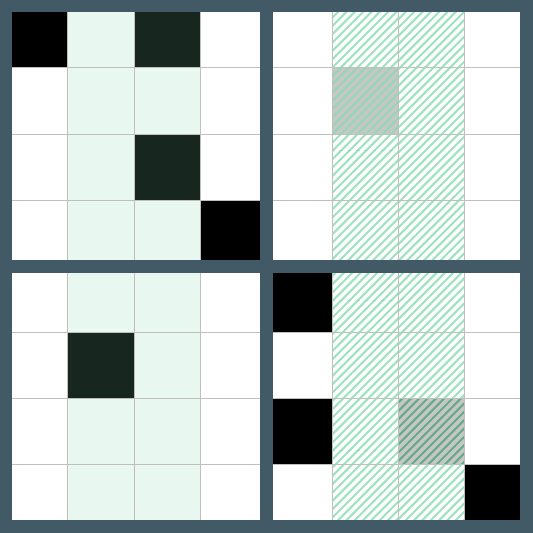}

\medskip

\noindent\rolloutblankarrow
\rolloutarrow{H_3}
\rolloutstate{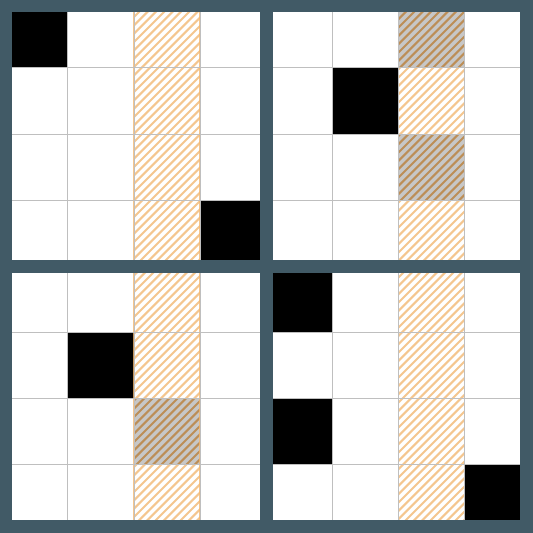}
\rolloutarrow{H_2}
\rolloutstate{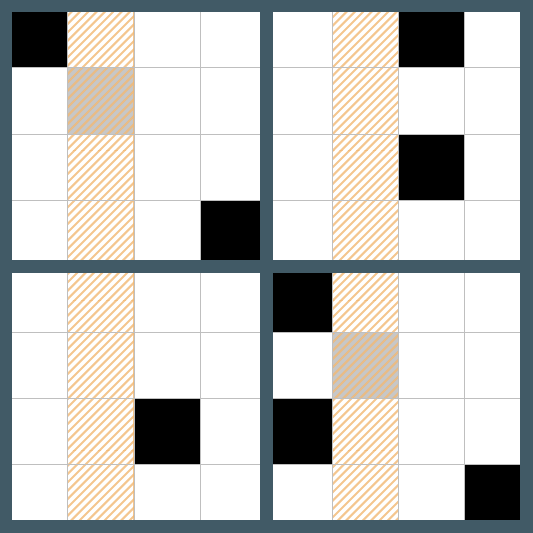}
\rolloutarrow{\mathrm{CZ}_{1,3}}
\rolloutstate{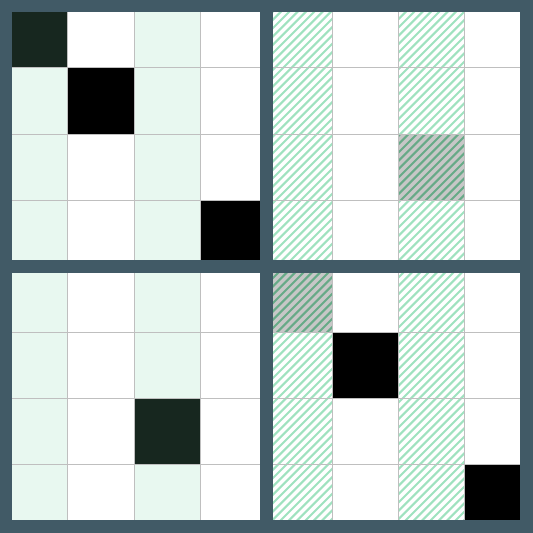}
\rolloutarrow{H_3}
\rolloutstate{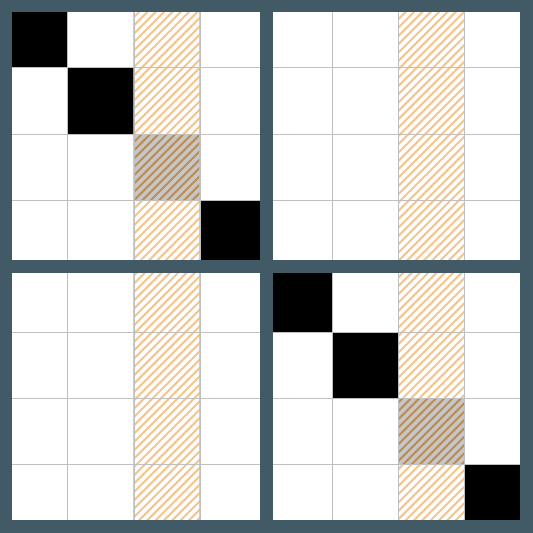}
\end{minipage}
\caption{Full reduction sequence for the four-qubit target tableau in Figure~\ref{fig:clifford-examples}. Each arrow is labelled by the next elementary gate in the displayed circuit, and the sequence ends at the identity tableau.}
\label{fig:appendix-clifford-tableau-rollout}
\end{figure*}

This appendix briefly reviews the \emph{circuit model} and the \emph{stabilizer formalism} underlying the synthesis problem studied in the main text. In the standard circuit model, a quantum algorithm is expressed as a sequence of elementary gates acting on qubits. A \emph{universal gate set} such as Clifford+$T$ can approximate arbitrary quantum computations to arbitrary precision. In this paper we focus on the Clifford subset, generated by $H$, $S$, and $\mathrm{CZ}$, because it admits a compact exact representation unavailable for general circuits. For a more complete treatment of Clifford circuits and stabilizer tableaus, see Aaronson and Gottesman \cite{AaronsonGottesman2004}.

The \emph{Pauli observables} correspond to the canonical measurement bases of a qubit. Their matrix representatives are
\[
I=
\begin{bmatrix}
1 & 0\\
0 & 1
\end{bmatrix},
\quad
X=
\begin{bmatrix}
0 & 1\\
1 & 0
\end{bmatrix},
\quad
Y=
\begin{bmatrix}
0 & -i\\
i & 0
\end{bmatrix},
\quad
Z=
\begin{bmatrix}
1 & 0\\
0 & -1
\end{bmatrix}.
\]
The \emph{Clifford generators} are the one-qubit gates
\[
H=\frac{1}{\sqrt{2}}
\begin{bmatrix}
1 & 1\\
1 & -1
\end{bmatrix},
\qquad
S=
\begin{bmatrix}
1 & 0\\
0 & i
\end{bmatrix},
\]
and the two-qubit \emph{controlled-$Z$ gate}
\[
\mathrm{CZ}=
\begin{bmatrix}
1 & 0 & 0 & 0\\
0 & 1 & 0 & 0\\
0 & 0 & 1 & 0\\
0 & 0 & 0 & -1
\end{bmatrix}.
\]

A Clifford circuit can be characterized equivalently as a circuit that maps \emph{Pauli operators} to Pauli operators under conjugation. This can be verified directly on the generators:
\[
H X H = Z,\qquad H Z H = X,
\]
\[
S X S^\dagger = Y,\qquad S Z S^\dagger = Z,
\]
and for the two-qubit gate $\mathrm{CZ}$,
\[
\begin{gathered}
\mathrm{CZ}(X \otimes I)\mathrm{CZ}^\dagger = X \otimes Z,\qquad
\mathrm{CZ}(Z \otimes I)\mathrm{CZ}^\dagger = Z \otimes I,\\[4pt]
\mathrm{CZ}(I \otimes X)\mathrm{CZ}^\dagger = Z \otimes X,\qquad
\mathrm{CZ}(I \otimes Z)\mathrm{CZ}^\dagger = I \otimes Z.
\end{gathered}
\]
Here $X_i$, $Y_i$ and $Z_i$ denote the corresponding Pauli operator on qubit $i$ and identity on every other qubit. A \emph{Pauli string} is a tensor product of single-qubit Pauli operators; for example, on two qubits $Y_1Z_2=Y\otimes Z$.

Since each generator maps Pauli operators to Pauli operators under conjugation, any circuit built from these generators does as well. It follows that a Clifford operation can be specified exactly by how it transforms the Pauli generators $X_1,\ldots,X_n,Z_1,\ldots,Z_n$ under conjugation. Those generator images determine the full Clifford action up to global phase. The only consistency requirement is that these images preserve the original commutation relations: $X_i$ commutes with $X_j$, $Z_i$ commutes with $Z_j$, and $X_i$ commutes with $Z_j$ if and only if $i \neq j$; on the same qubit, $X_i$ and $Z_i$ anticommute.

The binary tableau $M_U\in\mathbb{F}_2^{2n\times2n}$ records these generator images with their signs omitted. Each row records how an $n$-qubit Clifford operator $U$ transforms one of $X_1,\ldots,X_n,Z_1,\ldots,Z_n$ under conjugation. This binary matrix is the state used by our synthesis agent. To record the signs as well, one can append a column $r_U\in\mathbb{F}_2^{2n}$, giving the signed tableau $(M_U\mid r_U)$. The generator images are Hermitian, so their phases are signs $\pm1$, represented by $(-1)^r$.

We encode each Pauli string using two bits per qubit: $(x_k,z_k)=(0,0),(1,0),(0,1),(1,1)$ represent $I,X,Z,Y$, respectively. The first $n$ columns of $M_U$ contain the $X$ components and the last $n$ columns contain the $Z$ components. A signed row therefore has the form
\[
(x_1,\ldots,x_n\mid z_1,\ldots,z_n\mid r).
\]
For example, $(1,0\mid1,1\mid1)$ represents $-Y_1Z_2$. The upper and lower halves of the tableau track the images of the $X$ and $Z$ generators, respectively, while its left and right halves record their $X$ and $Z$ components. These two divisions give four $n\times n$ blocks. With the sign column separated on the right, the full tableau has the form
\begin{center}
\includegraphics{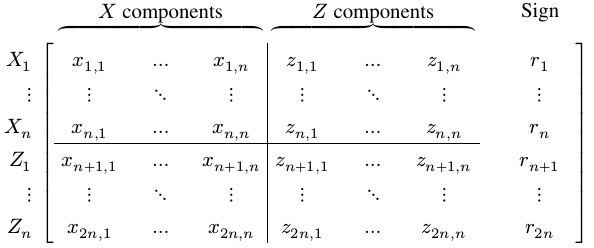}
\end{center}
The binary tableau figures elsewhere in this paper omit the sign column.

\paragraph{Examples.}
On one qubit, $H$ exchanges $X$ and $Z$, while $S$ maps $X$ to $Y$ and leaves $Z$ unchanged. Their signed tableaus are shown below, with the corresponding Pauli image written beside each row:
\begin{center}
\begin{tabular}{@{}c@{\qquad\qquad}c@{}}
\includegraphics{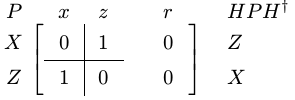} & \includegraphics{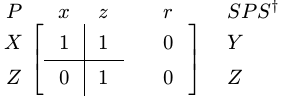}
\end{tabular}
\end{center}
For a two-qubit example, use the generator ordering $X_1,X_2,Z_1,Z_2$. The conjugation rules for $\mathrm{CZ}_{1,2}$ give
\[
X_1\mapsto X_1Z_2,\qquad
X_2\mapsto Z_1X_2,\qquad
Z_1\mapsto Z_1,\qquad
Z_2\mapsto Z_2,
\]
and hence
\begin{center}
\includegraphics{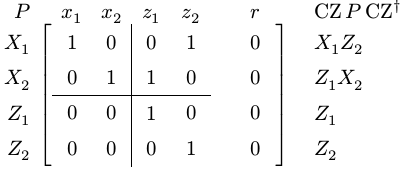}
\end{center}
For example, the first row $(1,0\mid0,1\mid0)$ represents $+X_1Z_2$: conjugation by $\mathrm{CZ}$ maps $X_1$ to $X_1Z_2$. Although all sign bits are zero for these generators, they need not remain zero when gates are composed. Applying $S$ twice maps $X$ to $-X$ and leaves $Z$ unchanged, so $M_{S^2}=I_2$ but $r_{S^2}=(1,0)^T$.

\paragraph{Commutation and the symplectic condition.}
The product of two phase-free Pauli strings corresponds to addition of their binary row vectors over $\mathbb{F}_2$. For instance, $XZ$ has binary encoding $(1,1)$, the same as $Y$ up to phase. To determine whether two Pauli strings commute, we pair the $X$ components of each with the $Z$ components of the other. For row vectors $u=(x\mid z)$ and $v=(x'\mid z')$, this gives the symplectic inner product
\[
u\Omega v^T=x\cdot z'+z\cdot x',
\qquad
\Omega=\begin{bmatrix}0&I_n\\I_n&0\end{bmatrix},
\]
where all arithmetic is modulo two. The result is zero when the Pauli strings commute and one when they anticommute. Equivalently, for Pauli strings $P$ and $Q$ with binary encodings $u$ and $v$,
\[
PQ=(-1)^{u\Omega v^T}QP.
\]
Conjugation preserves these relations: $X_i$ anticommutes with $Z_i$, while all other pairs of distinct generators commute. Since the generator images are stored as rows, their pairwise symplectic inner products give
\[
M_U\Omega M_U^T=\Omega,
\]
equivalently $M_U^T\Omega M_U=\Omega$, as used in the main text. There are no further constraints on the binary part of the tableau: valid binary Clifford actions are precisely the elements of $\Sp(2n,\mathbb{F}_2)$.

\paragraph{Why we omit the signs.}
The agent operates only on $M_U$, which determines the Clifford operation up to a Pauli correction and global phase. The omitted signs can be restored by single-qubit Pauli gates without increasing the two-qubit gate count. Thus reducing $M_U$ to the identity solves the phase-free synthesis problem studied here; if a signed target is specified, a final Pauli correction recovers its signs. The identity circuit has $M_I=I_{2n}$ and $r_I=0$.

\paragraph{Composition and gate updates.}
Our row convention fixes the order of tableau multiplication. Applying a gate $G$ after a circuit $U$ gives the operator $GU$, and each stored image $P'=UPU^\dagger$ is updated to $GP'G^\dagger$. On the binary tableau this is right multiplication:
\[
M_{GU}=M_U M_G.
\]
Thus each gate updates every row independently. The following rules update only the binary matrix $M_U$; sign-bit updates are omitted, and $\oplus$ denotes addition modulo two. For the one-qubit generators, the update to each row affects only the coordinates of the qubit on which the gate acts:
\[
H_i:\ (x_i,z_i) \mapsto (z_i,x_i),
\qquad
S_i:\ (x_i,z_i) \mapsto (x_i, x_i \oplus z_i),
\]
with all other qubit coordinates unchanged.

For the two-qubit generator $\mathrm{CZ}_{i,j}$, the binary update couples only the coordinates of qubits $i$ and $j$:
\[
\mathrm{CZ}_{i,j}:\quad
z_i \mapsto z_i \oplus x_j,\qquad
z_j \mapsto z_j \oplus x_i,
\]
with all other binary coordinates unchanged. Equivalently, appending $\mathrm{CZ}_{i,j}$ adds the $X$-support column of qubit $j$ into the $Z$-support column of qubit $i$, and vice versa.

These local rules describe tableau reduction: every action available to the agent is an indexed generator, and applying that action performs an exact sparse update of the current binary state. In this phase-free binary representation, all three generator tableaus are involutions, including $M_S^2=I$ over $\mathbb{F}_2$; this should be distinguished from the underlying quantum gate, where $S^2=Z$ rather than the identity. Reaching the identity tableau therefore solves the group-reduction problem described in the main text, after which reversing the applied generators recovers a circuit with the original target's binary tableau. A final Pauli correction supplies any required signs.

\begin{figure*}[t]
\centering
\includegraphics[width=0.98\textwidth]{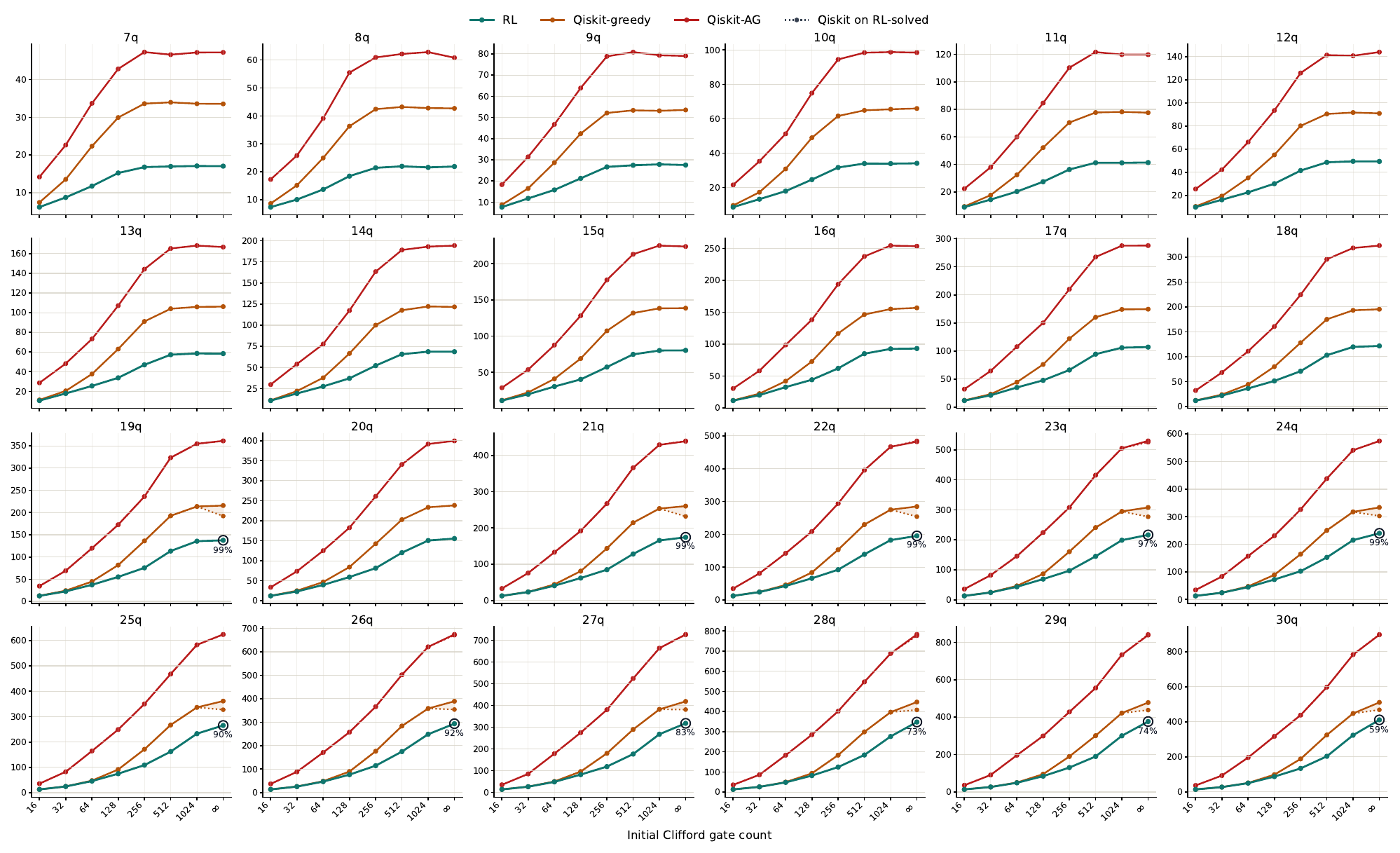}
\caption{Full CZ-count sweep for the ten-qubit-trained model, expanding the main-text view in Figure~\ref{fig:10fc-difficulty-grid} to every evaluation width from 7 to 30 qubits. Each panel fixes the evaluation qubit count and sweeps the displayed random-walk gate counts 16 through 1024, followed by the $\infty$ endpoint representing uniform Clifford sampling, under the same $6n^2$ rollout budget with the no-loop safeguard. Solid Qiskit curves show all-target means; dotted markers of the same colour and shaded gaps show the corresponding Qiskit means on the subset solved by the learned policy when these differ.}
\label{fig:appendix-10fc-difficulty-grid-full}
\end{figure*}

\begin{figure*}[t]
\centering
\includegraphics[width=0.98\textwidth]{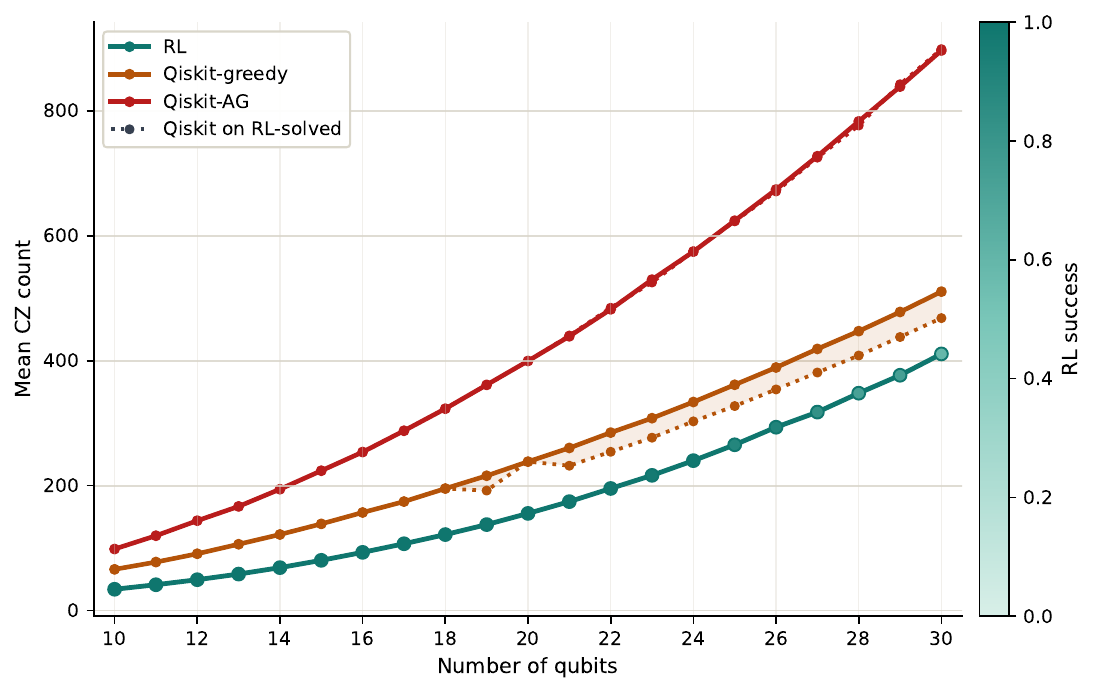}
\caption{Uniform-target CZ-count comparison for the ten-qubit-trained checkpoint under the same $6n^2$ rollout decoder with the no-loop safeguard. The benchmark uses the separate fully random Clifford target family, with $100$ held-out targets at each displayed qubit count from $10$ to $30$ qubits, and compares the learned policy against Qiskit's implementation of the Bravyi \textit{et al.} greedy synthesizer and Aaronson--Gottesman on the same target sets. Learned means are computed over solved instances only, learned-policy markers are coloured by solve rate, solid Qiskit curves show all-target means, and dotted segments of the same colour show Qiskit means on the subset solved by the learned policy.}
\label{fig:appendix-10fc-three-method-cz}
\end{figure*}

\begin{figure*}[t]
\centering
\includegraphics[width=0.98\textwidth]{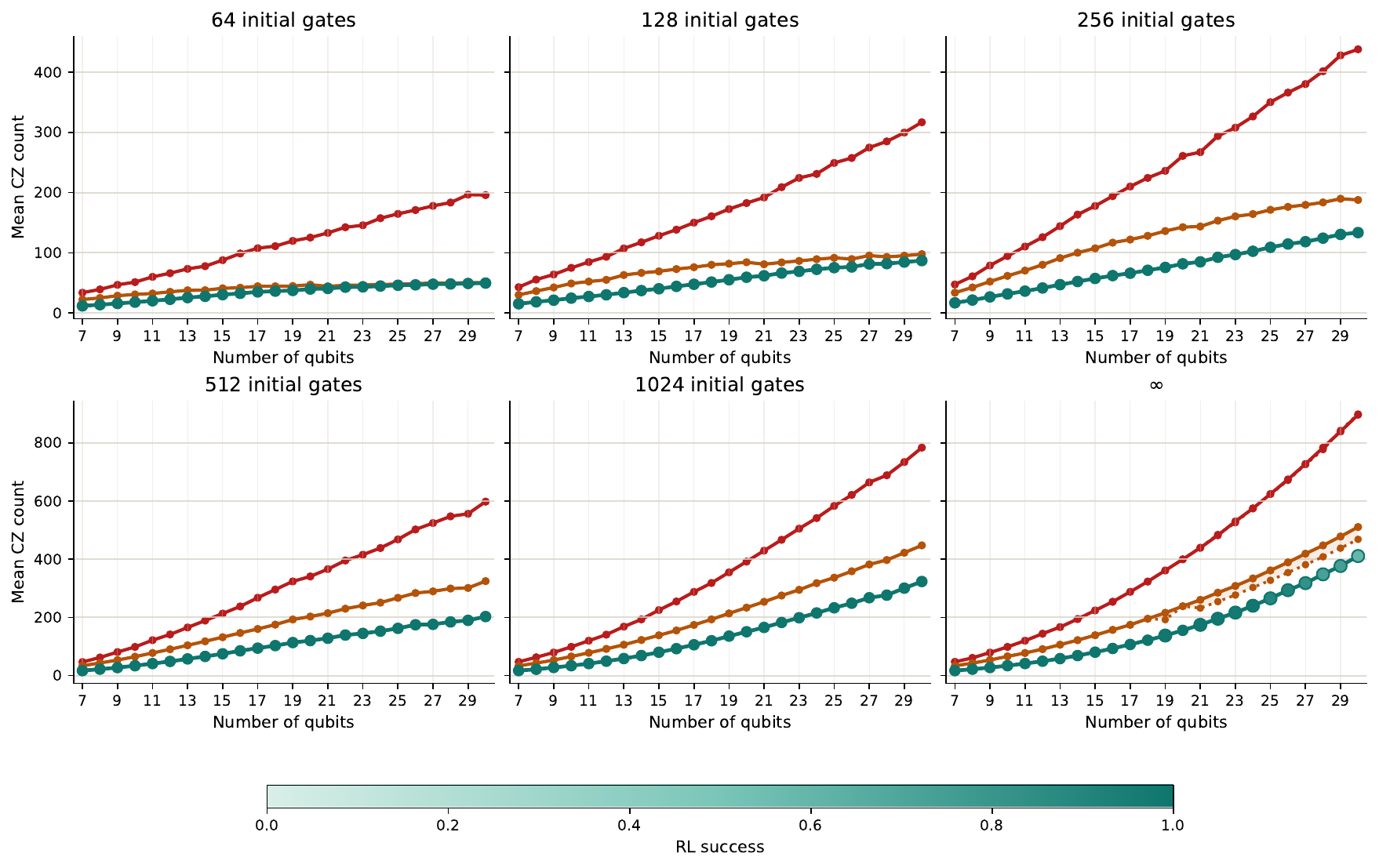}
\caption{CZ-count comparison for the ten-qubit-trained model across the harder displayed random-walk gate counts 64, 128, 256, 512, and 1024, together with the $\infty$ setting for uniform Clifford targets, under the same $6n^2$ rollout budget with the no-loop safeguard. Each panel plots mean CZ count against evaluation qubit count on the full sweep from 7 to 30 qubits. Solid Qiskit curves show all-target means; dotted markers of the same colour and shaded gaps show Qiskit means restricted to targets solved by the learned policy when it does not solve every target.}
\label{fig:appendix-10fc-three-method-total-full}
\end{figure*}

\newpage
\subsection{Six-qubit Checkpoint Extrapolation}
\label{app:sixfc-extrapolation}

Figure~\ref{fig:appendix-6fc-difficulty-grid-full} and Table~\ref{tab:appendix-sixfc-extrapolation} show the zero-shot behaviour of the checkpoint trained only on six-qubit targets. Under the same $6n^2$ no-loop rollout decoder used elsewhere, this checkpoint solves every larger-circuit target in the canonical sweep: all random-walk batches from $7$ to $30$ qubits and all fully random Clifford batches from $7$ to $30$ qubits. The cost is circuit quality rather than reliability. The six-qubit checkpoint remains strong on smaller and easier evaluation settings, often beating both Qiskit baselines in CZ count, but its CZ count rises sharply on high-depth and fully random targets with the most qubits. In contrast, the checkpoint further trained on ten qubits gives much lower CZ counts when it succeeds, but begins to fail on fully random targets beyond about $24$ qubits.

\begin{figure*}[t]
\centering
\includegraphics[width=0.98\textwidth]{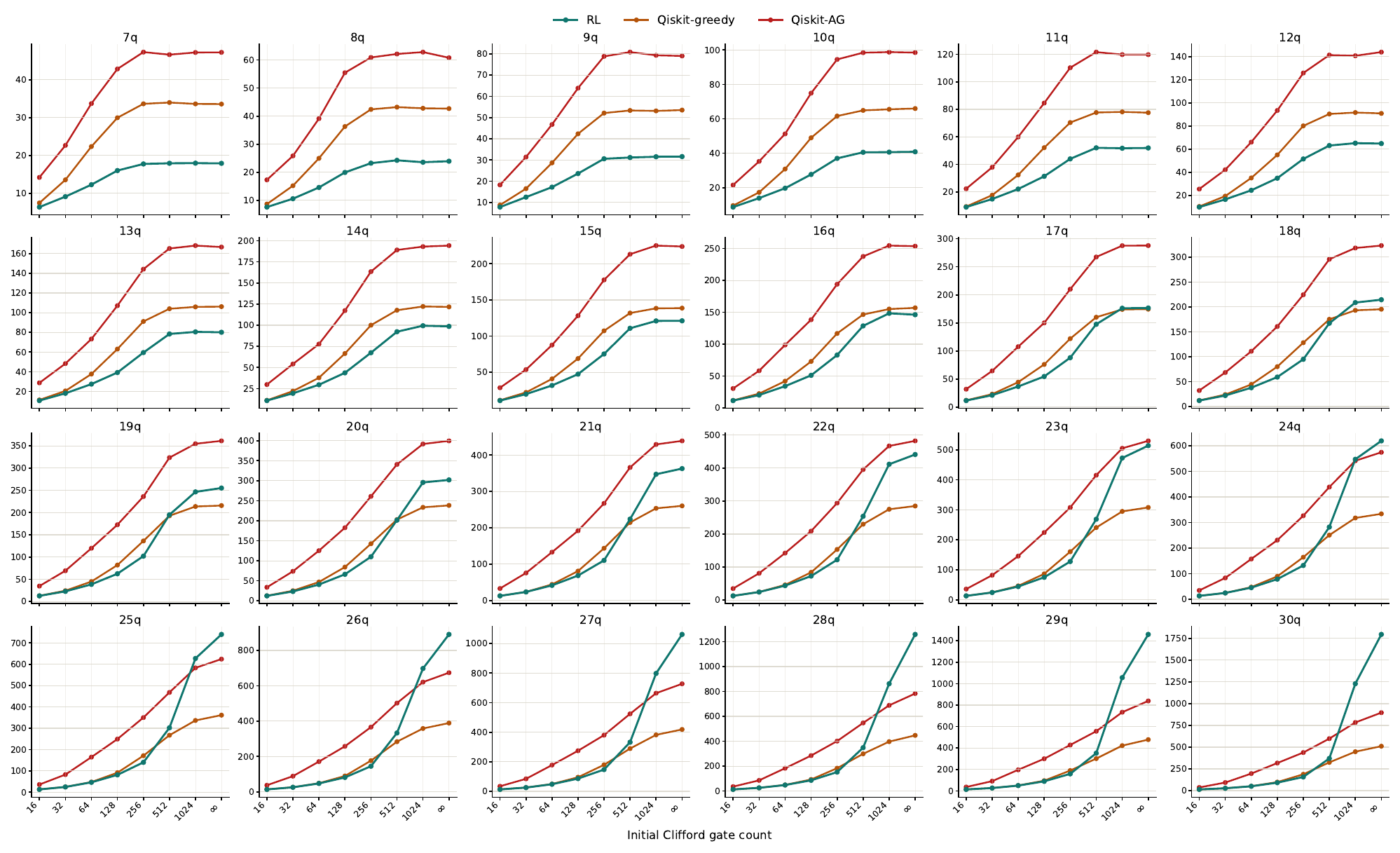}
\caption{Full CZ-count sweep for the six-qubit-trained model, using the same layout as Figure~\ref{fig:appendix-10fc-difficulty-grid-full}. Each panel fixes the evaluation qubit count and sweeps the displayed random-walk gate counts 16 through 1024, followed by the $\infty$ endpoint representing uniform Clifford sampling, under the same $6n^2$ rollout budget with the no-loop safeguard. The checkpoint solves every target shown; the figure therefore directly compares mean CZ count against the two Qiskit baselines on the same target batches.}
\label{fig:appendix-6fc-difficulty-grid-full}
\end{figure*}

\begin{table*}[t]
\centering
\scriptsize
\setlength{\tabcolsep}{3.5pt}
\begin{tabular}{lrrrrrr}
\toprule
Evaluation subset & Cells & Succ. & 6q model CZ & AG CZ & Greedy CZ & 6q model lower CZ than AG / greedy\\
\midrule
Std. $7$--$20$ q. & 98 & 1.000 & 52.97 & 114.06 & 65.14 & 98/98 \; / \; 93/98\\
Std. $21$--$30$ q., $d\leq128$ & 40 & 1.000 & 41.43 & 134.81 & 46.59 & 40/40 \; / \; 40/40\\
Std. $21$--$30$ q., $d\geq256$ & 30 & 1.000 & 383.56 & 479.53 & 263.19 & 23/30 \; / \; 10/30\\
Fully random & 24 & 1.000 & 448.68 & 388.72 & 229.33 & 17/24 \; / \; 10/24\\
\bottomrule
\end{tabular}
\caption{Zero-shot extrapolation of the six-qubit-trained checkpoint on larger circuits. Each cell is one matched setting of target family, qubit count, and difficulty from the canonical evaluation package. Entries report mean CZ count over solved instances, and the six-qubit trained model solves all targets. The final row uses the fully random Clifford target family, which is tracked separately from the standard random-walk sweep.}
\label{tab:appendix-sixfc-extrapolation}
\end{table*}

\begin{table*}[p]
\centering
\small
\begin{tabular}{llccc}
\toprule
Difficulty & Family & Success rate & Mean CZ & Seeds\\
\midrule
$d=20$ & Ours & 1.000 & $6.03 \pm 0.01$ & 3 \\
$d=20$ & RelTransformer~\cite{Shaw2018RelativeAttention} & 1.000 & $6.04 \pm 0.01$ & 3 \\
$d=20$ & Transformer~\cite{Vaswani2017Attention} & 1.000 & $6.06 \pm 0.03$ & 3 \\
$d=20$ & MLP & 1.000 & $6.19 \pm 0.01$ & 3 \\
$d=20$ & FlatMLP & 1.000 & $6.14 \pm 0.01$ & 3 \\
\midrule
$d=100$ & Ours & 1.000 & $11.99 \pm 0.04$ & 3 \\
$d=100$ & RelTransformer~\cite{Shaw2018RelativeAttention} & 1.000 & $12.06 \pm 0.03$ & 3 \\
$d=100$ & Transformer~\cite{Vaswani2017Attention} & 1.000 & $12.50 \pm 0.03$ & 3 \\
$d=100$ & MLP & 1.000 & $13.89 \pm 0.12$ & 3 \\
$d=100$ & FlatMLP & 1.000 & $13.62 \pm 0.16$ & 3 \\
\midrule
$d=1024$ & Ours & 1.000 & $13.71 \pm 0.06$ & 3 \\
$d=1024$ & RelTransformer~\cite{Shaw2018RelativeAttention} & 1.000 & $13.79 \pm 0.15$ & 3 \\
$d=1024$ & Transformer~\cite{Vaswani2017Attention} & 1.000 & $14.37 \pm 0.17$ & 3 \\
$d=1024$ & MLP & 1.000 & $16.46 \pm 0.09$ & 3 \\
$d=1024$ & FlatMLP & 1.000 & $15.67 \pm 0.15$ & 3 \\
\midrule
$d=\infty$ & Ours & 1.000 & $13.49 \pm 0.01$ & 3 \\
$d=\infty$ & RelTransformer~\cite{Shaw2018RelativeAttention} & 1.000 & $13.47 \pm 0.06$ & 3 \\
$d=\infty$ & Transformer~\cite{Vaswani2017Attention} & 1.000 & $14.12 \pm 0.28$ & 3 \\
$d=\infty$ & MLP & 1.000 & $15.88 \pm 0.05$ & 3 \\
$d=\infty$ & FlatMLP & 1.000 & $15.19 \pm 0.16$ & 3 \\
\bottomrule
\end{tabular}
\caption{Full six-qubit architecture ablation over three seeds on fixed 100-target validation sets per difficulty. Finite $d$ is the initial Clifford gate count used to generate each random-walk target; $d=\infty$ uses uniformly sampled Clifford targets. Entries report mean CZ count over solved instances under rollout decoding, as mean $\pm$ standard deviation across seeds; all completed-model success rates are $1.000$.}
\label{tab:appendix-clifford-architecture-ablation}
\end{table*}

\section{Decoding Details}
\label{app:decoding-details}

Given a trained policy and a target tableau $M_{\mathrm{target}}$, greedy decoding repeatedly applies the generator with the highest policy probability to the current tableau, stopping when the identity is reached or when the rollout budget is exhausted.
The resulting reduction sequence is reversed to obtain a circuit for $M_{\mathrm{target}}$, as in the reverse-reduction formulation of Section~\ref{sec:method}.
We also use the inverse-tableau trick of Gidney~\cite{Gidney2020InverseTableaus}: decode both $M_{\mathrm{target}}$ and $M_{\mathrm{target}}^{-1}$, convert the inverse-tableau solution back to a circuit for $M_{\mathrm{target}}$, and keep the shorter successful circuit.
For sampled policy-guided rollouts, we repeat the same reduction procedure under a fixed compute budget and retain the shortest successful circuit found.
In the larger-qubit sweeps, rollouts use a no-loop safeguard that rejects actions returning to a tableau already visited in the current trajectory when an alternative action is available.

All learned-policy targets and success checks use the phase-free binary symplectic tableau. For the $\infty$ target family generated with Qiskit's \texttt{random\_clifford}, we store and evaluate only the returned \texttt{symplectic\_matrix}; Pauli signs are not part of the target state. The Qiskit baselines are compared on the same target batches at the level of two-qubit cost: CZ and CX gates each count as one entangling gate, SWAP counts as three, and single-qubit, Pauli, and phase-correction gates are excluded from the reported CZ-equivalent counts.

For the six-qubit Bravyi \textit{et al.} benchmark in Section~\ref{sec:experiments}, Table~\ref{tab:appendix-six-qubit-search} gives the decoding settings behind the headline exact-recovery curve. The greedy pass evaluates the target and inverse tableaus once and keeps the shorter successful circuit. The longer search uses the same target/inverse comparison, a rollout cap of 512 generator applications, the native CUDA rollout backend, and 2048 parallel rollout environments.

It combines results from four temperature schedules: \texttt{t40}, \texttt{t40$\to$greedy}, \texttt{t40$\to$t25}, and \texttt{t40$\to$t12}, where \texttt{t40} samples at fixed temperature $4.0$, \texttt{t40$\to$t12} samples with temperature linearly decreasing from $4.0$ to $1.2$, and \texttt{t40$\to$greedy} decreases from $4.0$ to $0.05$. The fixed-temperature \texttt{t40} arm alone reaches $990/1003$ exact recoveries in the recorded first-hit trace; the annealed arms contribute the final five recoveries. Each arm runs 4096 trials per circuit for each of the target and inverse tableaus, so one arm corresponds to $1003\times4096\times2=8{,}216{,}576$ sampled reductions.

The reported milestone times record when each recovery count was first reached across the four schedules: once an optimal circuit is found for a benchmark instance, later samples can still improve other instances, but the first-hit timestamp for that instance is fixed.

\begin{table*}[p]
\centering
\small
\resizebox{\textwidth}{!}{%
\begin{tabular}{lrrrrl}
\toprule
Decoder stage & Exact / total & Time & Cumulative reductions & Step cap & Temperature schedules from $4.0$\\
\midrule
Greedy prefilter & 507/1003 & 0.83 s & 2,006 & 512 & greedy\\
Brief search pass & 610/1003 & 21 s & -- & 512 & sampled policy rollouts\\
First reaches prior SOTA & 982/1003 & 21.6 min & 3,410,999 & 512 & $4.0$\\
First reaches 990 & 990/1003 & 51.9 min & 8,170,010 & 512 & $4.0$\\
First reaches final count & 995/1003 & 183.1 min & 29,143,574 & 512 & four schedules: $4.0$, greedy, $2.5$, $1.2$\\
\bottomrule
\end{tabular}
}
\caption{Decoding settings for the six-qubit 1003-circuit benchmark. ``Reductions'' counts complete sampled reduction attempts, including both target and inverse tableaus when inverse comparison is enabled. Timings are total suite time, not per-instance time, from the recorded native CUDA rollout traces on a single NVIDIA A100 80GB GPU. The brief search pass brings every target to within one CZ gate of optimality, as reported in the main text. The full first-hit trace records the combined results of the four temperature schedules.}
\label{tab:appendix-six-qubit-search}
\end{table*}

\begin{table}[p]
\centering
\small
\begin{tabular}{lr}
\hline
Hyperparameter & Value\\
\hline
Optimizer & PPO\\
Training steps & $1.5\times 10^9$\\
Parallel environments & 2048\\
Rollout length & 256\\
Curriculum range & difficulties 1--1000\\
Curriculum advance threshold & success $=1.00$\\
Learning rate & $2.5\times 10^{-4}$\\
Batch size & 8192\\
Epochs per update & 5\\
Discount factor $\gamma$ & 0.99\\
GAE parameter $\lambda$ & 0.95\\
Policy clip coefficient & 0.15\\
Value clip coefficient & 0.2\\
\texttt{hamming\_left} scale & 0.5\\
Single-qubit gate cost & 0.01\\
Terminal bonus & 25\\
\hline
\end{tabular}
\caption{Training hyperparameters for the main 6-qubit Clifford run.}
\label{tab:appendix-training-hparams}
\end{table}

\section{Reproducibility, Compute, and Assets}
\label{app:reproducibility}

The accompanying \href{https://github.com/y-richie-y/cliffords-data}{code and data release} contains the code, trained checkpoints, fixed target batches, recorded result tables, figure-generation scripts, and a README with exact commands for the local evaluator, ablation validation, ensemble validation, the 1003-circuit benchmark, and the basic PPO training script. The bundled evaluator can run on CPU, with optional GPU acceleration for larger sweeps. The headline six-qubit search timings are recorded in the supplied first-hit trace and summary tables; the larger-qubit evaluation tables record the elapsed times of learned rollouts and Qiskit baselines for each setting. A live browser demo for the released six-qubit model, together with the recovered Bravyi \textit{et al.} benchmark circuits, is hosted at \href{https://y-richie-y.github.io/clifford/}{\texttt{y-richie-y.github.io/clifford/}}.

The work uses synthetic Clifford targets and public benchmark data from Bravyi \textit{et al.}; no human-subject or scraped personal data are used. Existing software and benchmark assets are cited where used, and the code release identifies the bundled source trees, checkpoint files, and benchmark CSVs needed to reproduce the reported experiments.

The expected positive impact is improved Clifford synthesis, which may reduce two-qubit gate counts in quantum-compilation workflows. Possible negative impacts are indirect: better circuit synthesis can marginally reduce the cost of quantum experiments, including experiments whose downstream applications are outside the scope of this paper. The released assets are specialized to Clifford synthesis and do not contain sensitive data or general-purpose generative models.

\clearpage
\section*{NeurIPS Paper Checklist}

\begin{enumerate}
\item {\bf Claims}
  \item[] Question: Do the main claims made in the abstract and introduction accurately reflect the paper's contributions and scope?
  \item[] Answer: \answerYes{}
  \item[] Justification: The abstract and introduction state the benchmark sizes, success rates, and timing scope used in the experiments; for example, $99.2\%=995/1003$ on the six-qubit benchmark.

\item {\bf Limitations}
  \item[] Question: Does the paper discuss the limitations of the work performed by the authors?
  \item[] Answer: \answerYes{}
  \item[] Justification: Section~\ref{sec:experiments} and the Discussion section describe the reliability--quality tradeoff, lack of optimality guarantees beyond six qubits, and scaling limitations of the action space.

\item {\bf Theory assumptions and proofs}
  \item[] Question: For each theoretical result, does the paper provide the full set of assumptions and a complete (and correct) proof?
  \item[] Answer: \answerYes{}
  \item[] Justification: The random-walk limit theorem and the equivariant-policy result state their assumptions and are proven in Appendices~\ref{app:random-walk-limit} and~\ref{app:equivariant-optimal-policy}.

\item {\bf Experimental result reproducibility}
  \item[] Question: Does the paper fully disclose all the information needed to reproduce the main experimental results of the paper to the extent that it affects the main claims and/or conclusions of the paper (regardless of whether the code and data are provided or not)?
  \item[] Answer: \answerYes{}
  \item[] Justification: The method, target generation, decoding protocol, hyperparameters, and evaluation settings are specified in the paper and appendix; the supplement provides fixed targets, checkpoints, result tables, and commands for the main evaluations.

\item {\bf Open access to data and code}
  \item[] Question: Does the paper provide open access to the data and code, with sufficient instructions to faithfully reproduce the main experimental results, as described in supplemental material?
  \item[] Answer: \answerYes{}
  \item[] Justification: The supplementary reproducibility bundle includes the local evaluator, PPO source tree, model checkpoints, target batches, benchmark CSVs, analysis tables, figure scripts, environment setup, and reproduction commands.

\item {\bf Experimental setting/details}
  \item[] Question: Does the paper specify all the training and test details (e.g., data splits, hyperparameters, how they were chosen, type of optimizer) necessary to understand the results?
  \item[] Answer: \answerYes{}
  \item[] Justification: Section~\ref{sec:method} describes the MDP, reward, curriculum, and architecture; Section~\ref{sec:experiments} specifies target generation and evaluation batches; Appendix~\ref{app:decoding-details} and Table~\ref{tab:appendix-training-hparams} give decoding and training details.

\item {\bf Experiment statistical significance}
  \item[] Question: Does the paper report error bars suitably and correctly defined or other appropriate information about the statistical significance of the experiments?
  \item[] Answer: \answerYes{}
  \item[] Justification: The six-qubit benchmark reports exact counts over the full 1003-circuit benchmark, larger sweeps report solve rates and matched 100-target means, and the architecture ablation reports mean $\pm$ standard deviation over three seeds.

\item {\bf Experiments compute resources}
  \item[] Question: For each experiment, does the paper provide sufficient information on the computer resources (type of compute workers, memory, time of execution) needed to reproduce the experiments?
  \item[] Answer: \answerYes{}
  \item[] Justification: Appendix~\ref{app:reproducibility} and the supplement state that evaluations run on CPU with optional CUDA, provide recorded wall-clock timings for the headline search and sweeps, and give commands for reproducing the lightweight and full evaluations.

\item {\bf Code of ethics}
  \item[] Question: Does the research conducted in the paper conform, in every respect, with the NeurIPS Code of Ethics \url{https://neurips.cc/public/EthicsGuidelines}?
  \item[] Answer: \answerYes{}
  \item[] Justification: The work uses synthetic Clifford targets and public benchmark data, does not involve human subjects or personal data, and releases specialized scientific software and checkpoints.
\item {\bf Broader impacts}
  \item[] Question: Does the paper discuss both potential positive societal impacts and negative societal impacts of the work performed?
  \item[] Answer: \answerYes{}
  \item[] Justification: Appendix~\ref{app:reproducibility} describes the positive compilation impact and the indirect dual-use risk that improved quantum compilation may reduce the cost of downstream quantum experiments.

\item {\bf Safeguards}
  \item[] Question: Does the paper describe safeguards that have been put in place for responsible release of data or models that have a high risk for misuse (e.g., pre-trained language models, image generators, or scraped datasets)?
  \item[] Answer: \answerNA{}
  \item[] Justification: The released assets are specialized Clifford-synthesis code, synthetic targets, and small policy checkpoints, not high-risk generative models or sensitive datasets.

\item {\bf Licenses for existing assets}
  \item[] Question: Are the creators or original owners of assets (e.g., code, data, models), used in the paper, properly credited and are the license and terms of use explicitly mentioned and properly respected?
  \item[] Answer: \answerYes{}
  \item[] Justification: Existing benchmark data, Qiskit synthesizers, PufferLib, and related methods are cited in the manuscript; the supplement includes asset and license notes identifying bundled source trees, benchmark CSVs, and upstream dependencies.

\item {\bf New assets}
  \item[] Question: Are new assets introduced in the paper well documented and is the documentation provided alongside the assets?
  \item[] Answer: \answerYes{}
  \item[] Justification: The supplement documents the new checkpoints, generated target batches, result tables, figure data, and evaluation scripts, with README instructions and smoke checks.

\item {\bf Crowdsourcing and research with human subjects}
  \item[] Question: For crowdsourcing experiments and research with human subjects, does the paper include the full text of instructions given to participants and screenshots, if applicable, as well as details about compensation (if any)?
  \item[] Answer: \answerNA{}
  \item[] Justification: The work does not involve crowdsourcing or human-subject experiments.

\item {\bf Institutional review board (IRB) approvals or equivalent for research with human subjects}
  \item[] Question: Does the paper describe potential risks incurred by study participants, whether such risks were disclosed to the subjects, and whether Institutional Review Board (IRB) approvals (or an equivalent approval/review based on the requirements of your country or institution) were obtained?
  \item[] Answer: \answerNA{}
  \item[] Justification: The work does not involve human subjects.

\item {\bf Declaration of LLM usage}
  \item[] Question: Does the paper describe the usage of LLMs if it is an important, original, or non-standard component of the core methods in this research? Note that if the LLM is used only for writing, editing, or formatting purposes and does \emph{not} impact the core methodology, scientific rigor, or originality of the research, declaration is not required.
  \item[] Answer: \answerNA{}
  \item[] Justification: The core methods described in the paper do not use LLMs as part of the research contribution.
\end{enumerate}

\end{document}